\documentclass[twoside]{IEEEtran}

\usepackage{amsfonts,amssymb}
\usepackage[usenames,dvipsnames,svgnames,table]{xcolor}

\usepackage[font=small,labelfont=bf]{caption}
\usepackage{vkmacros}
\graphicspath{{figures/}}

\usepackage{tabu}
\newif\ifmapx
\mapxfalse

\usepackage{environ}

\NewEnviron{apxonly}
{
 \ifmapx\expandafter
 \addtocounter{equation}{-1}
\begin{subequations}\renewcommand\theequation{\theparentequation\alph{equation}}
  \begingroup\color{violet!80!black}
 \BODY
 \endgroup
\end{subequations}
 \fi
}

\long\def\apxonl#1{\ifmapx{\color{violet!80!black}#1}\fi}

\newif\ifshowtodo
\showtodofalse
\long\def\todo#1{\ifshowtodo{\color{red!80!black}TODO: #1}\fi}

\newif\ifshowdetail
\showtodofalse
\long\def\detail#1{\ifshowdetail{\color{green!50!black}#1}\fi}

\psfrag{Rd}{\small{$R(d)$}}
\psfrag{n}{\small{$n$}}
\psfrag{R}{\small{$R$}}

\title{Data compression with low distortion \\and finite blocklength}

\author{
Victoria Kostina, \IEEEmembership{Member, IEEE}  \\
\thanks{
V. Kostina is with California Institute of Technology (e-mail: \href{mailto:vkostina@caltech.edu}{vkostina@caltech.edu}). 
This research is supported in part by the National Science Foundation (NSF)
under Grant CCF-1566567, and by the Simons Institute for the Theory of Computing. 
Parts of this paper were presented at the Allerton conference \cite{kostina2015allerton,kostina2016slb}. 
}
}

\IEEEoverridecommandlockouts
\begin{document}
\maketitle
\begin{abstract}

This paper considers lossy source coding of $n$-dimensional memoryless sources and shows an explicit approximation to the minimum source coding rate required to sustain the probability of exceeding distortion $d$ no greater than $\epsilon$, which is simpler than known dispersion-based approximations.  

Our approach takes inspiration in the celebrated classical result stating that the Shannon lower bound to rate-distortion function becomes tight in the limit $d \to 0$. We formulate an abstract version of the Shannon lower bound that recovers both the classical Shannon lower bound and the rate-distortion function itself as special cases. Likewise, we show that a nonasymptotic version of the abstract Shannon lower bound recovers all previously known nonasymptotic converses.

 A necessary and sufficient condition for the Shannon lower bound to be attained exactly is presented. It is demonstrated that whenever that condition is met, 
the rate-dispersion function is given simply by the varentropy of the source. Remarkably, all finite alphabet sources with balanced distortion measures satisfy that condition in the range of low distortions.  

Most continuous sources violate that condition. Still, we show that lattice quantizers closely approach the nonasymptotic Shannon lower bound, provided that the source density is smooth enough and the distortion is low. This implies that fine multidimensional lattice coverings are nearly optimal in the rate-distortion sense even at finite $n$.  The achievability proof technique is based on a new bound on the output entropy of lattice quantizers in terms of the differential entropy of the source, the lattice cell size and a smoothness parameter of the source density. The technique avoids both the usual random coding argument and the simplifying assumption of the presence of a dither signal.  

\end{abstract}


\section{Introduction}

We showed in \cite{kostina2011fixed} that for the compression of a stationary memoryless source under a single-letter distortion measure, the minimum achievable source coding rate $R(n,d,\epsilon)$ comparable with blocklength $n$ and the probability $\epsilon$ of exceeding distortion $d$ is given by
\begin{equation}
R(n,d,\epsilon) =  R(d) + \sqrt{\frac {\mathcal V(d)} n }\Qinv{\epsilon} + \bigo{ \frac{\log n} n }, \label{block2order}
\end{equation}
where $Q$ is the complementary Gaussian cdf, $R(d)$ is the rate-distortion function, and $\mathcal V(d)$ is the rate-dispersion function of the source. The rate-dispersion function quantifies the overhead over the rate-distortion function incurred by the finite blocklength constraint. 
Dropping the $\bigo{ \frac{\log n} n }$ remainder term in \eqref{block2order}, we obtain a simple approximation to the minimum achievable coding rate. That approximation provides good accuracy even at short blocklengths, as evidenced by the numerical results in \cite{kostina2011fixed}. 
 
In this contribution, we derive a simplification of \eqref{block2order} in the regime of low $d$, which corresponds to the practically relevant regime of high resolution data compression. 
The interest in pursuing such a simplification stems from the fact that 
closed-form formulas for $R(d)$ and $\mathcal V(d)$ are rarely available. 
Indeed, as shown in \cite{kostina2011fixed}, both the rate-dispersion and the rate-distortion function are described parametrically in terms of the solution to the rate-distortion convex minimization problem, defined for a source distribution $P_X$ and a distortion measure $\sd \colon \mathcal X  \times \mathcal Y \mapsto \mathbb R_+$  as
\begin{equation}
\mathbb R_{ X}(d) \triangleq \inf_{\substack{P_{Y|X}\colon \mathcal X \mapsto \mathcal Y \\ \E{\sd(X, Y)} \leq d}} I(X; Y) \label{eq:RR(d)}.
\end{equation} 
The rate-distortion and the rate-dispersion function are given by the mean and the variance of the $\mathsf d$-tilted information, the random variable that is defined for a distribution $P_X$ and a distortion measure $\sd$ as
 \begin{equation}
\jmath_{X}(x, d) \triangleq \log \frac 1 {\E{ \exp\left\{ \lambda^\star d - \lambda^\star \sd(x, Y^\star)\right\}}} \label{eq_idball}, 
\end{equation}
where $\lambda^\star = - \mathbb R^\prime_X(d)$ is the negative of the slope of the rate-distortion function of $X$ at distortion $d$, and the expectation is with respect to the unconditional distribution of $Y^\star$, the random variable that attains the minimum in \eqref{eq:RR(d)}, i.e. $\mathbb R_{ X}(d)  = I(X; Y^\star)$. \footnote{We assume for now that such a random variable exists.} Although the convexity of the problem in \eqref{eq:RR(d)} often allows for an efficient numerical computation of its optimum \cite{blahut1972computation}, closed-form expressions are available only in special cases. In those cases, the distortion measure is carefully tailored to match the source.  

The absence of an explicit expression for the $\mathsf d$-tilted information motivates a closer look into the behavior of  \eqref{eq_idball}. This paper shows that under regularity conditions and as long as $d$ is small enough, the $\sd$-tilted information in a random variable $X \in \mathbb R^n$ is closely approximated by, with high probability, 
\begin{equation}
 \jmath_X(X, d) \approx \log \frac 1 {f_X(X)} - \phi(d)  \label{eq:jappx},
\end{equation}
 where $f_X$ is the source density, and $\phi(d)$ is a term that depends only on the distortion measure and distortion threshold $d$. For example, for the mean-square error (MSE) distortion, 
\begin{equation}
\phi(d) = n \log \sqrt{2 \pi e d}. 
\end{equation}

If the source alphabet is finite and all columns of the
distortion matrix $\{\sd(x,y)\}_{x, y}$ consist of the same set of entries (balanced distortion measure) an even stronger claim holds, namely, 
\begin{equation}
 \jmath_X(X, d) = \log \frac 1 {P_X(X)} - \phi(d)  \qquad \text{a.s.}, \label{eq:jfinite}
\end{equation}
as long as $d \leq d_c$, where $d_c > 0$ is a function of $P_X$ and the distortion measure only.

The value of $\jmath_{X}(x, d)$ can be loosely interpreted as the amount of information that needs to be stored about $x$ in order to restore it with distortion $d$ \cite{kostina2011fixed}.
The explicit nature of \eqref{eq:jappx} illuminates the tension between the likelihood of $x$ and the target distortion: the likelier realization $x$ is, the fewer bits are required to store it; the lower tolerable $d$ is, the more bits are required in order to represent the source with that distortion.\footnote{$\phi(d)$ is strictly increasing in $d$.} This intuitively pleasing insight is not afforded by the general formula \eqref{eq_idball}.

To gain further understanding of the form of \eqref{eq:jappx}, recall that the Shannon lower bound \cite{shannon1959coding} states that the rate-distortion function is bounded below by the difference between the differential entropy of the source and a term that depends only on the distortion measure and distortion threshold $d$: 
\begin{equation}
 \mathbb R_X(d) \geq  \ushort {\mathbb R}_X(d) = h(X) - \phi(d), \label{slb}
\end{equation}
where $h(X)$ is the differential entropy of the source. Due to its simplicity and because it becomes increasingly tight in the limit of low distortion \cite{linkov1965evaluation,linder1994asymptotic}, the Shannon lower bound is often used as a convenient proxy for ${\mathbb R}_X(d)$. The statement in \eqref{eq:jappx} can be viewed as a nonasymptotic refinement of those results. 
More precisely, this paper proposes a nonasymptotic version of the Shannon lower bound, valid at any $d$, and demonstrates that at low $d$, the bound can be approached by a lattice quantizer followed by a lossless coder.  A careful analysis of those bounds reveals that for a class of difference distortion measures and stationary memoryless sources with sufficiently smooth  densities, as $d \to 0$ and $n \to \infty$, the nonasymptotically achievable source coding rate admits the following expansion: 
\begin{equation}
 R(n,d,\epsilon) =  \ushort{R}(d) + \sqrt{\frac {\ushort{\mathcal V}} n }\Qinv{\epsilon} + \bigo{\sqrt d} + \bigo{ \frac{\log n} n }, \label{eq:Rapprox}
\end{equation}
where $\ushort{R}(d) = \ushort{\mathbb R}_{\sX}(d)$ is Shannon's lower bound for $P_{\sX}$, the single-letter distribution of the source, and  
\begin{equation}
 \ushort{\mathcal V} \triangleq \Var{ \log f_\sX(\sX) } \label{vlb}.
\end{equation}
Thus, similar to \eqref{block2order}, $\ushort{R}(d)$ and  $ \ushort{\mathcal V} $ are given by the mean and the variance of the random variable on the right side of \eqref{eq:jappx} (particularized to $P_{\sX}$). The term $\bigo{\sqrt d}$ in \eqref{eq:Rapprox}, which is always nonnegative, is the penalty due to the density of the source not being completely flat within each lattice quantization cell. Naturally, this term vanishes as the sizes of quantization cells decrease, and its magnitude depends on smoothness of the source density.   

Since \eqref{eq:Rapprox} is attained by lattice quantization,  lattice quantizers are nearly optimal at high resolution even at finite blocklength. The implication for engineering practice is that, in a search for good codes, it is unnecessary to consider more complex structures than lattices if the goal is high resolution digital representation of the original analog signal. Due to the regularity of
the code vector locations, lattice quantizers offer a great reduction in the complexity of encoding
algorithms (e.g. \cite{conway1982fast,gersho2012vector}). Therefore, both their performance and their regular algebraic structure make lattices a particularly appealing choice for an efficient analog-to-digital conversion.

This paper also develops new results on the rate-distortion performance of lattice quantization of continuous sources with memory. 
We prove that for a class of sources satisfying a smoothness condition, variable-length lattice quantization attains Shannon's lower bound in the limit of $n \to \infty$ and $d \to 0$, even if the source is nonergodic or nonstationary.   Furthermore, if the source density is log-concave, we show that Shannon's lower bound is attained at a speed $\bigo{\frac 1 {\sqrt n}}$ with increasing blocklength, providing the first result of this sort for lossy compression of sources with memory. 

The key to our study of lattice quantization is an explicit nonasymptotic bound on the probability distribution (and, in particular, the entropy) observed at the output of a lattice quantizer for $X$. The bound is a function of the lattice cell size and a smoothness parameter of the source density. The bound provides an estimate of the speed of convergence in the classical results by R\'enyi \cite[Theorem 4]{renyi1959dimension} and 
 Csisz\'ar \cite{csiszar1971entropyquantization}. 

Another essential ingredient  of our development is a new, abstract formulation of the Shannon lower bound that encompasses the classical Shannon lower bound as a special case and that does not impose any symmetry conditions on the distortion measure. An appropriate choice of an auxiliary measure in the abstract Shannon lower bound makes it equal to the rate-distortion function.   Likewise, a nonasymptotic version of the abstract Shannon lower bound recovers all previously known nonasymptotic converses. 

 We state necessary and sufficient conditions for the abstract Shannon lower bound to hold with equality.  In particular, those conditions allow us to establish the validity of \eqref{eq:jfinite} for finite alphabet sources with balanced distortion. Inserting \eqref{eq:jfinite} into \eqref{block2order}, we conclude that for discrete memoryless sources with balanced distortion, for all $d \leq d_c$, the nonasymptotic fundamental limit is given simply by
\begin{equation}
 R(n,d,\epsilon) =  \ushort{R}(d) + \sqrt{\frac {\ushort {\mathcal V}} n }\Qinv{\epsilon} + \bigo{ \frac{\log n} n },\label{eq:Rapproxf}
\end{equation}
 which is a sharpened version of \eqref{eq:Rapprox} without the $\bigo{\sqrt d}$ term.  More generally, \eqref{eq:Rapproxf} holds for any stationary memoryless source whose rate-distortion function meets the Shannon lower bound with equality. 


Notable prior contributions to the understanding of lattice quantizers in large dimension include the works by Rogers \cite{rogers1964packing}, Gersho \cite{gersho1979asymptotically}, Zamir and Feder \cite{zamir1992universal} and Linder and Zeger \cite{linder1994tessellating}. Rogers \cite{rogers1964packing} showed the existence of efficient lattice coverings of space. Using a heuristic approach, Gersho \cite{gersho1979asymptotically} studied tessellating vector quantizers, i.e. quantizers whose regions are congruent with some tessellating convex polytope $P$.\footnote{A polytope $P$ is called {\it tessellating} if there exists a partition of $\mathbb R^n$ consisting of translated and rotated copies of $P$.} Although every lattice quantizer is a tessellating quantizer, the converse is not true.
Gersho \cite{gersho1979asymptotically} showed heuristically that in the limit of low distortion, tessellating vector quantizers approach $n$-dimensional Shannon's lower bound to within a term of order $\bigo{\frac {\log n}{n}}$. Relying on a conjecture by Gersho,  Linder and Zeger \cite{linder1994tessellating} streamlined the proof of Gersho's result and reported that the minimum entropy among all $n$-dimensional tessellating vector quantizers approaches the $n$-letter Shannon's lower bound in the limit of low $d$, provided that the Gersho conjecture is true. Zamir and Feder \cite{zamir1992universal} considered the setting in which a
signal called a dither is added to an input signal prior to quantization, namely, dithered quantization, and showed an upper bound on the achievable conditional (on the dither) output entropy of dithered lattice quantizers. Their result relied on a rather restrictive assumption on the source density violated even by the Gaussian distribution. That assumption was later relaxed by Linder and Zeger \cite{linder1994tessellating}. Zamir and Feder \cite{zamir1996onlatticenoise} went on to study the properties of a vector uniformly distributed over a lattice quantization cell, and showed that the normalized second moment of the optimal lattice quantizer approaches that of a ball, $\frac 1 {2 \pi e}$, as the dimension increases. While the assumption of the availability of the dither signal both at the encoder and the decoder greatly simplifies the analysis and improves the performance somewhat by smoothening the underlying densities, it can also substantially complicate the engineering implementation. This paper does not consider dithered quantization. 

Historically, theoretical performance analysis of lossy compressors proceeded in two disparate directions: bounds derived from Shannon theory \cite{shannon1959probability}, and bounds derived from high resolution approximations \cite{bennett1948spectra,zador1982quantizationerror}. The former provides asymptotic results as the sources are coded using longer and longer blocks. The latter assumes fixed input block size but estimates the performance as the encoding rate becomes arbitrarily large. This paper fuses the two approaches to study the performance of high resolution block compressors from the Shannon theory viewpoint.

So as not to clutter notation, in those statements in which the Cartesian structure of the space is unimportant, we will denote random vectors simply by $X$, $Y$, etc., omitting the dimension $n$. Wherever necessary, we will make the dimension explicit, writing $X^n$, $Y^n$ in lieu of $X$, $Y$.  For a stationary memoryless process $X_1, X_2, \ldots$, we denote by $\sX$ the random variable that is distributed the same as $X_i, i = 1,2, \ldots$. All logarithms are arbitrary common base.  The Euclidean norm is denoted by $\| \cdot \|$ and the $L_p$-norm is denoted by $\| \cdot \|_p$.


\apxonl{
Gerrish \cite{gerrish1964nongaussian} showed that the Shannon
lower bound is attainable if and only if a process can be decomposed
into two statistically independent processes with one
of them being a white Gaussian process

Blahut's \cite{blahut1972computation} comment on the computation of $R(d)$ in the non-discrete case: the search for extremizing probability distributions is
now a problem in the calculus of variations with constraints,
but otherwise closely follows the discrete case.
}

The rest of the paper is organized as follows. \secref{sec:slb} presents the abstract Shannon lower bound together with its nonasymptotic counterpart, shows a necessary and sufficient condition for the abstract Shannon lower bound to be attained with equality and demonstrates \eqref{eq:jfinite}.   \secref{sec:lattice} presents a nonasymptotic (in both $n$ and $d$) analysis of lattice quantization and shows an upper bound on the output entropy of lattice quantizers. \secref{sec:asymp} presents an asymptotic analysis of the bounds in  \secref{sec:slb} and \secref{sec:lattice}.  For a class of sources with memory, \secref{sec:asymp1} shows that lattice quantization attains the Shannon lower bound in the limit of large $n$ and small $d$. For stationary memoryless sources, \secref{sec:asymp2} presents a refined asymptotic analysis that quantifies how fast that limit is approached and establishes \eqref{eq:Rapprox}.  For clarity of presentation, the exposition in \secref{sec:lattice} and \secref{sec:asymp} is restricted to the MSE distortion. The generalization to non-MSE distortion measures is postponed until Section \ref{sec:beyondMSE}.

\section{The abstract Shannon lower bound and its nonasymptotic counterpart}
\label{sec:slb}
 \subsection{The abstract Shannon lower bound}

The (informational) rate-distortion function in \eqref{eq:RR(d)} admits the following parametric representation. 

\begin{thm} [{Parametric representation of $\mathbb R_X(d)$, \cite{csiszar1974extremum}}]
 \label{sc:thm:csiszarg}
Suppose that the following basic assumptions are satisfied. 
\begin{enumerate}[(a)]
 
  \item  \label{sc:item:a}
 $\mathbb R_X(d)$ is finite for some $d$, i.e. $ d_{\min} < \infty$, where
\begin{equation}
 d_{\min} \triangleq \inf \left\{ d\colon ~ \mathbb R_X(d) < \infty \right\} \label{sc:dmin}.
\end{equation}

 \item \label{sc:item:c} The distortion measure is such that there exists a finite set $E \subset \mathcal Y$ such that
\begin{equation}
 \E{ \min_{y \in E} \mathsf d(X, y)} < \infty \label{sc:dcsiszar}.
\end{equation}
\end{enumerate}

For each $d > d_{\min}$, it holds that
\begin{equation}
 \mathbb R_X(d) = \max_{g(x), ~ \lambda}\left\{- \E{ \log  {g(X)} } - \lambda d\right\} \label{eq_RR(d)csiszar},
\end{equation}
where the maximization is over $g(x) \geq 0$ and $\lambda\geq 0$ satisfying the constraint
\begin{equation}
\E{ \frac {\exp\left(  - \lambda \mathsf d(X, y)\right) } {g(X)} } \leq 1 ~ \forall y \in \mathcal Y \label{eq_csiszarg}.
\end{equation}
\end{thm}

\begin{remark}
 The maximization over $g(x) \geq 0$ in \eqref{eq_RR(d)csiszar} can be restricted to only $0 \leq g(x) \leq 1$\cite{csiszar1974extremum}. Equality in \eqref{eq_csiszarg} holds for  $P_{Y^\star}$-a.s. $y$.   
 \apxonl{
 Csisz\'ar showed that the maximum in \eqref{eq_RR(d)csiszar} remains the same if $0 \leq g(x) \leq 1$ is allowed. Assume that the maximizer $0 \leq g^\star(x) \leq 1$ attains  the maximum. Letting  $g(x) = a > 1$ for $x \in A$, where the Lebesgue measure of $A$ is nonzero, and $g(x) = g^\star(x)$ everywhere else. While $g(x)$ still satisfies \eqref{eq_csiszarg}, the value of $- \E{ \log  {g(X)} }$ can only decrease.  
 }
 \label{rem:1}
\end{remark}

\begin{remark}
The $\sd$-tilted information (defined in \eqref{eq_idball},  \cite{kostina2011fixed}) can be alternatively defined as
\begin{equation}
 \jmath_X(x, d) = - \log g(x) - \lambda^\star d, \label{eq:jg}
\end{equation}
where the pair $(g(\cdot), \lambda)$ attains the maximum in \eqref{eq_RR(d)csiszar}. So, 
\begin{equation}
\mathbb R_X(d) = \E{\jmath_X(X, d)}. \label{eq:Ej} 
\end{equation}
Furthermore, if the infimum in \eqref{eq:RR(d)} is attained by some $Y^\star$, i.e. $\mathbb R_{X}(d) = I(X; Y^\star)$, then
\begin{equation}
g(x) =  \E{ \exp\left( - \lambda^\star \sd(x, Y^\star)\right)}
\end{equation}
leads to the definition in \eqref{eq_idball}.  
\label{rem:2}
\end{remark}

\apxonl{
TODO: add prefix and non-prefix converses. 
}

For finite alphabet sources,  a parametric representation of $\mathbb R_X(d)$  is contained in Shannon's paper \cite{shannon1959coding}; both Gallager's \cite[Theorem 9.4.1]{gallager1968information} and Berger's \cite{berger1971rate} texts contain parametric representations of $\mathbb R_X(d)$ for discrete and continuous sources. However, it was Csisz\'ar \cite{csiszar1974extremum} who gave rigorous proofs of \eqref{eq_RR(d)csiszar} in the following much more general setting: $X$ belongs to a general abstract probability space, and the existence of the conditional distribution $P_{Y^\star | X}$ attaining $\mathbb R_X(d)$ is not required.

Here, we leverage the result of Csisz\'ar to state a generalization of the Shannon lower bound to abstract probability spaces. 

Each choice of $\lambda \geq 0$   and $g$ satisfying \eqref{eq_csiszarg} gives rise to a lower bound to $\mathbb R_X(d)$. The Shannon lower bound corresponds to a particular choice of $(\lambda, g)$. 

Let $\mu$ be a measure on $\mathcal X$ such that the distribution of $X$ is absolutely continuous with respect to $\mu$. Denote the density of the distribution of $X$ with respect to $\mu$ (Radon-Nikodym derivative) by
\begin{equation}
f_{X \| \mu}(x) \triangleq \frac{dP_X}{d \mu} (x) \label{eq:density},
\end{equation}
and the corresponding log-likelihood ratio by
\begin{align}
\imath_\mu(x) \triangleq \log \frac{d\mu}{d P_X}(x). 
\end{align}

The differential entropy with respect to $\mu$ can be defined as
\begin{align}
h_\mu(X) &\triangleq  \E{\imath_\mu(X)} \label{eq:hmu}\\
&= - D(f_X \| \mu).
\end{align}

If $X$ is a continuous random variable, a natural choice for $\mu$ is the Lebesgue measure. Then, the density in \eqref{eq:density} is known as the probability density function, and $h_\mu(X)$ is simply $h(X)$, the differential entropy of $X$.  If $X$ is a discrete random variable, a natural choice for $\mu$ is the counting measure. Then, the density in \eqref{eq:density} is the probability mass function, and $h_\mu(x)$ is equal to $H(X)$, the entropy of $X$.

It is easy to verify that the choice of $\lambda$ and $g$ in Table \ref{tab:slb} satisfies \eqref{eq_csiszarg}.  The Shannon lower bound can now be generalized to abstract spaces and arbitrary distortion measures as follows. 
\begin {table}
\centering
\resizebox{.8 \linewidth}{!}
{%
\begin{tabu}{  c  }
$
\begin{aligned}
\Sigma \triangleq &~ \sup_{y\in \mathcal Y} \int \exp(-\lambda \sd(x, y) ) d\mu(x) \notag\\
=&~ \int \exp(-\lambda \sd(x, y_\lambda) ) d\mu(x) \label{eq:sup}
\end{aligned} $\\ 
\hline
$
\dfrac{ dP_{X|Y^\star = y}}{d\mu}(x) \triangleq \dfrac {\exp(-\lambda \mathsf d(x, y))} {\int \exp(-\lambda \sd(x, y) ) d\mu(x) }
$\\
\hline
 $
 \begin{aligned}
 \phi_\mu(d) &\triangleq  \log \Sigma + \lambda d\\
 g(x) &= f_{X \| \mu}(x)  \Sigma \\
 \end{aligned}
 $  \\
\hline
$\lambda > 0$: arbitrary 
\\
\hline
\end{tabu}
}
\caption {The choice of $(g(x), \lambda)$ in \eqref{eq_RR(d)csiszar} that leads to the abstract Shannon's lower bound in Theorem \ref{thm:slb}. } 
\label{tab:slb} 
\end{table}


\begin{thm}[the abstract Shannon lower bound]
 Fix a measure $\mu$ such that the distribution of $X$ is absolutely continuous with respect to $\mu$. 
 For all $d > d_{\min}$, 
 \begin{equation}
\mathbb R_X(d) \geq h_\mu(X) - \phi_\mu(d). \label{eq:slbgen}
\end{equation} 
\label{thm:slb}
\end{thm}


 Theorem \ref{thm:slb} provides a family of lower bounds parameterized by the choice of base measure $\mu$. In the classical Shannon lower bound, $\mu$ is a Lebesgue measure (or a counting measure, if the alphabet is discrete) and the distortion measure satisfies a symmetry condition, so that the integral in the definition of $\Sigma$ in Table \ref{tab:slb} does not depend on the choice of $y$. 
 Shannon's original derivation  \cite{shannon1959probability} applied to continuous sources under the mean-square error distortion, and it did not use a parametric representation of $\mathbb R_X(d)$.   A decade later, Pinkston \cite{pinkston1969application} derived a version of the bound for finite alphabet sources with a distortion measure such that all the columns of the per-letter distortion matrix $\{\sd(x,y)\}_{x, y}$ consist of the same set of entries. A generalization of the discrete Shannon lower bound to distortion measures not satisfying any symmetry conditions was put forth by Gray \cite{gray1971stationary}. The bound in \thmref{thm:slb} is more general than these results and recovers them as special cases. 
 
The right-side of \eqref{eq:slbgen} can be made equal to $\mathbb R_X(d)$ by choosing $\mu$ to satisfy: 
\begin{align}
\imath_\mu(x) 
 \label{eq:muopt} 
&=  \jmath_X(x, d). 
\end{align}
To verify that the choice of $\mu$ in \eqref{eq:muopt} results in equality in \eqref{eq:slbgen}, observe using \eqref{eq:hmu} that 
\begin{equation}
 h_\mu(X) = \E{\jmath_X(X, d)},
\end{equation}
and that
\begin{align}
 \phi_\mu(d) &= \log \Sigma + \lambda d\\
 &= \sup_{y \in \mathcal Y} \log \E{ \exp\left(  - \lambda \mathsf d(X, y) + \jmath_{X}(X, d) \right) } + \lambda d\\
 &= 0, \label{eq:phiopt}
\end{align}
where to obtain \eqref{eq:phiopt} we used \eqref{eq_csiszarg}, \eqref{eq:jg} and \remref{rem:1}.

The long-standing appeal of the Shannon lower bound is that one can obtain a tight bound on the rate-distortion function even without the knowledge of the distribution that attains it, as \eqref{eq:muopt} demands. For an illustration of such a calculation, suppose that $\mathcal X$ is a set endowed with a group operation, $``+''$, satisfying the group axioms. Then, it makes sense to talk about $x + y$ and $x - y = x + (-y)$, where $-y$ is the inverse of $y$  (according to the group operation). 
Distortion measures of the form
\begin{equation}
\sd(x, y)  = \sd( x - y ) \label{eq:difference}
\end{equation}
are called {\it difference} distortion measures. If $\mathcal X = \mathbb R^n$ and $\sd$ is a difference distortion measure, 
then letting $\mu$ be the Lebesgue measure, we obtain
\begin{equation}
\Sigma = \int \exp(-\lambda \sd(x - y) ) dx, 
\end{equation}
regardless of the choice of $y$. So, we may set $y = 0$, and obtain the classical Shannon lower bound:
\begin{align}
\mathbb R_X(d) &\geq
\ushort{\mathbb R}_X(d) 
= h(X) - \phi(d) \label{eq:uR}
\end{align}
  -  see Table \ref{tab:slbdiff} for the definition of $\phi(d)$. 
In the same fashion, if $\mathcal X$ is a discrete group, letting $\mu$ be the counting measure on $\mathcal X$, we notice that
\begin{equation}
\Sigma = \sum_{x \in \mathcal X} \exp(-\lambda \sd(x - y) ), 
\end{equation}
for all $y \in \mathcal X$. Therefore, we may let $y = 0$ (the identity element of group $\mathcal X$) and obtain Pinkston's variant of the Shannon lower bound \cite{pinkston1969application}.  See Table \ref{tab:slbdiff}.

Throughout the paper, we denote by 
\begin{align}
\ushort {\mathbb R}(d) = \E{\ushort{\jmath}_{X}(X, d)} \label{eq:Rbar}
\end{align}
 the classical Shannon's lower bound, obtained with the choice of the auxiliary measure $\mu$ in Table \ref{tab:slbdiff}. The random variable $\ushort{\jmath}_{X}(X, d)$ is also defined in Table \ref{tab:slbdiff}. 


\begin {table}
\centering
\resizebox{1 \linewidth}{!}{%
\begin{tabu}{  c |  c }
$\mathcal X$ is discrete group, $X \in \mathcal X$ & $X \in \mathbb R^n$ is continuous \\
\hline 

$\Sigma =  \sum_{z \in \mathcal X} \exp(-\lambda \sd(z) )$
&
$
\Sigma =  \int_{\mathbb R^n} \exp(-\lambda \sd(z) ) dz
$\\

\hline
  $
  \begin{aligned}
  P_{Z_\lambda}(z) =   \frac{ \exp\left(-\lambda \sd(z)\right)}{\sum_{z \in \mathcal X} \exp\left(-\lambda \sd(z)\right)}
  \end{aligned} 
  $ & $f_{Z_\lambda}(z) = \dfrac{\exp(-\lambda \mathsf d(z))}{ \int_{\mathbb R^n} \exp(-\lambda \mathsf d(z)) dz}$  \\
  \hline 
  
  $  \begin{aligned} \phi(d) &= H(Z_{ \lambda}) \\
&= \log \sum_{z \in \mathcal X} \exp\left( \lambda d - \lambda \sd(z) \right)\\
 \end{aligned}$
 &
  $\begin{aligned} \phi(d) &=  h(Z_{ \lambda})\\& = \log \int_{\mathbb R^n} \exp\left(   \lambda d -  \lambda \sd(z) \right) dz \end{aligned}$ 
\\
\hline
$g(x) = P_X(x)  \exp\left( \phi(d) - \lambda d  \right)$ & $g(x) = f_X(x)  \exp\left( \phi(d) - \lambda d  \right)$\\
\hline
$
\ushort{\jmath}_{X}(x, d) \triangleq  \log \dfrac 1 {P_X(x)} - \phi(d)
$
&
$
\ushort{\jmath}_{X}(x, d) \triangleq   \log \dfrac 1 {f_X(x)} - \phi(d)
$\\
\hline
\multicolumn{2}{c}{$\lambda > 0$: solution  to equation $\E{\sd(Z_{ \lambda} )} = d$} \\
\hline
\end{tabu}
}
\caption {In the case where $\sd$ is difference distortion measure, the classical Shannon lower bound is obtained by letting the base measure $\mu$ in Table \ref{tab:slb} to be the counting measure if $\mathcal X$ is a discrete group, and the Lebesgue measure if $\mathcal X = \mathbb R^n$. 
} 
\label{tab:slbdiff} 
\end{table}

The calculation in Table \ref{tab:slbdiff} tacitly assumes that there exists solution $\lambda > 0$ to 
\begin{equation}
 \E{Z_\lambda} = d. \label{eq:lambda}
\end{equation}
An important question is under what conditions this solution exists.  
For continuous $X$, Linkov \cite[Lemma 1]{linkov1965evaluation} showed that \eqref{eq:lambda}
has a unique solution for all sufficiently small $d$, as long as $\mathsf d(\cdot)$  satisfies the following mild regularity conditions: 
\begin{eqenum}
 \item $\sd(r) = 0$ only at $r = 0$, and $\sd(r)$ is nondecreasing. \label{lin:a}
 \item There exists such $\nu > 0$ that $\lim_{r \to 0} r^{-\nu} \sd(r) < \infty$.
 \item $\int_{\mathbb R_+} \sd^2(r) \exp(- \sd(r)) dr < \infty$. \label{lin:c}
\end{eqenum}

For discrete sources, we show that if
\begin{equation}
 \sd(0) = 0, \quad \sd(z) > 0, z \neq 0. \label{eq:d00}
\end{equation}
and $|\mathcal X| = m$, then \eqref{eq:lambda} has a solution for all $ d \in (0, \E{\sd(Z_0)}]$. Indeed, 
observe using \eqref{eq:d00} that 
\begin{align}
\E{\sd(Z_0)} &= \frac 1 m \sum_{z \in \mathcal X} \sd(z),\\
\lim_{\lambda \to + \infty} \E{\sd(Z_\lambda)} &= 0. 
\end{align}
Since $\E{\sd(Z_\lambda)}$ is continuous as a function of $\lambda$ on $[0, + \infty)$, it follows that \eqref{eq:lambda} has a solution for all $ d \in (0, \E{\sd(Z_0)}]$. \apxonl{What about uniqueness?}

We proceed to list several examples of the calculation of the Shannon lower bound for difference distortion measures. 
\begin{example}
In the special case of $X \in \mathbb R^n$ and mean-square error distortion, 
we recover the classical Shannon lower bound \cite{shannon1959probability} as follows. 
Let $\mathsf d$ be the mean-square error distortion:
\begin{equation}
\mathsf d(x, y) = \frac 1 n \|x - y\|^2_2.  \label{eq:mse}
\end{equation}
A straightforward calculation using Table \ref{tab:slbdiff} reveals that,
\begin{align}
\lambda &= \frac {n}{2d} \log e,\\
\phi(d) &=  n \log \sqrt{2 \pi e d},
\end{align}
so if $X$ is a continuous real-valued random vector of length $n$, 
\begin{equation}
\ushort{\mathbb R}_X(d) =  h(X)  - n \log \sqrt{ 2 \pi e d}   \label{slb2}.
\end{equation}
\end{example}

 \begin{example}
For weighted mean-square error distortion measure,  
 \begin{equation}
\mathsf d(x, y) =  \frac 1 n \|  \mathsf W(x - y) \|_2^2 \label{eq:dnormMSEW},
\end{equation}
where $\mathsf W$ is an invertible $n \times n$ matrix, Shannon's lower bound is given by
\begin{equation}
\ushort{\mathbb R}_X(d) = h(X)  -n \log \sqrt{2\pi e d} + \log |\det \mat W|. 
\end{equation}
\end{example}

\begin{example}
Let $\mathsf d$ be the scaled $L^p$ norm distortion:
 \begin{equation}
\mathsf d(x, y) = n^{- \frac s p} \| x - y \|^s_p \label{eq:dLps},
\end{equation}
where $s > 0$. 
A direct calculation using Table \ref{tab:slbdiff} shows that Shannon's lower bound  is given by 
\begin{align}
\ushort{\mathbb R}_X(d) &= h(X) + \frac n s \log \frac 1 d - \frac n p \log n  - \log b_{n, p} \notag \\
 & + \frac n s \log \frac{n }{ s e} - \log \Gamma\left( \frac n s + 1\right),
\end{align}
where $b_{n, p}$ is the volume of a unit $L_p$ ball: 
\begin{equation}
b_{n, p} \triangleq \frac{ \left( 2 \Gamma \left(\frac 1 p +1\right) \right)^n }{ \Gamma \left(\frac n p+1 \right) }. \label{eq:bnp}
\end{equation}
\end{example}

\begin{example}
Assume that the alphabet is finite, $|\mathcal X| = m$, and consider the symbol error distortion
\begin{equation}
\sd(z) = 1\{ z = 0\}. 
\end{equation}
Then, 
\begin{equation}
\ushort{\mathbb R}_X(d) =  H(X) - h(d) - d \log(m - 1). \label{eq:rfinite}
\end{equation}
\end{example}

\subsection{The nonasymptotic abstract Shannon lower bound}
As it turns out, the abstract Shannon lower bound in Theorem \ref{thm:slb} has a nonasymptotic kin expressed in terms of the Neyman-Pearson function. 

The optimal performance achievable among all randomized tests $P_{W|X}\colon A \rightarrow \left\{ 0, 1\right\}$ between measures $P$ and $Q$ on $A$ is denoted by ($1$ indicates that the test chooses $P$):
\begin{equation}
\label{eq_beta}
\beta_{\alpha}(P, Q) = \min_{\substack{P_{W|X}\colon \\ P[W = 1] \geq \alpha}}  Q \left[ W = 1\right].
\end{equation}
Note that the Neyman-Pearson function $\beta_{\alpha}(P, Q)$ is well defined even if $P$ and $Q$ are not probability measures. 

An $(M, d, \epsilon)$ fixed-length lossy compression code is a pair of mappings $\mathsf f\colon \mathcal X \mapsto \{1, \ldots, M\}$ and $\mathsf g \colon \mathcal X \mapsto \{1, \ldots, M\}$, such that 
\begin{equation}
\Prob{\sd(X, \mathsf g(\mathsf f (X) ) ) > d} \leq \epsilon. 
\end{equation}

\begin{thm}
\label{thm:C}
Let $P_X$ be the source distribution defined on the alphabet $\mathcal X$. Any $(M,d,\epsilon)$ code must satisfy, for any measure $\mu$ on $\mathcal X$:
\begin{enumerate}[a)]
\item  \begin{equation}
\label{eq:C}
M \geq  \frac {\beta_{1 - \epsilon}(P_X, \mu)}{\sup_{y \in \mathcal Y}  \mu \left[ \sd(X, y) \leq d\right]}.
\end{equation}

\item \begin{align}
\epsilon 
\geq&~  \sup_{\gamma > 0} \left\{ \Prob{ \imath_\mu(X) - \phi_\mu(d) \geq \log M + \gamma} - \exp(-\gamma)\right\}, \label{eq:Ca} 
\end{align}

\end{enumerate}
\end{thm}

\begin{proof}
The inequality in \eqref{eq:C} is due to \cite[Theorem 8]{kostina2011fixed}. To show \eqref{eq:Ca}, note that for all $\zeta > 0$ (e.g. \cite{polyanskiy2012notes}), 
\begin{align}
 \zeta \beta_{1 - \epsilon}(P_X, \mu) \geq \Prob{ \imath_\mu(X) \geq - \log \zeta} - \epsilon \label{eq:betalb}.
\end{align}
On the other hand, by Markov's inequality, the $\mu$-volume of the distortion $d$-ball is linked to $\phi_\mu(d)$ as follows. 
\begin{align}
 \mu \left[ d(X, y) \leq d\right] &=\int d\mu(x) 1\left\{ \sd(x, y) \leq d \right\}\\
 &\leq \int d\mu(x) \exp(\lambda d - \lambda \sd(x, y)) \\
 &\leq \sup_{y \in \mathcal Y} \int d\mu(x) \exp(\lambda d - \lambda \sd(x, y)) \label{eq:markov}\\
 &= \exp (\phi_\mu(d)). \label{eq:markova}
\end{align}

Applying \eqref{eq:betalb} and \eqref{eq:markova} to \eqref{eq:C}, we conclude that for all $\zeta > 0$,
\begin{equation}
\epsilon \geq \Prob{\imath_\mu(X) \geq - \log \zeta} -\zeta M \exp(\phi_\mu(d)) . \label{eq:C3}
\end{equation}
Re-parameterizing \eqref{eq:C3} through
\begin{align}
 \zeta = \frac{1}{M \exp(\phi_\mu(d) + \gamma)},
\end{align}
we obtain \eqref{eq:Ca}. 
\end{proof}

As clear from the proof of \thmref{thm:C}, the bound in \eqref{eq:Ca} is a weakening of the bound in \eqref{eq:C}. 

Note the striking parallels between \thmref{thm:C} and the abstract Shannon lower bound in \thmref{thm:slb}. Both bounds require a choice of the base measure $\mu$. The optimal binary hypothesis test in \eqref{eq:C} is a function of the log-likelihood ratio $\imath_\mu(X)$ only, whose expectation is equal to $h_\mu(X)$, the first term in \eqref{eq:slbgen}. Furthermore, the denominator in the right side of \eqref{eq:C} is linked to  $\phi_\mu(d)$,  the second term in \eqref{eq:slbgen}, through \eqref{eq:markova}. 

The similarities between between \thmref{thm:slb} and \thmref{thm:C} become even more apparent if we look at the bound in \eqref{eq:Ca}. In a typical usage of \eqref{eq:Ca}, $\gamma$ is chosen so that  its contribution to the both terms in \eqref{eq:Ca} is negligible, and thus the excess-distortion probability is bounded through the distribution of $\imath_\mu(X) - \phi_\mu(d)$ as
\begin{equation}
 \epsilon \gtrapprox  \Prob{ \imath_\mu(X) - \phi_\mu(d) \geq \log M}. \label{eq:Capprox}
\end{equation}

As discussed above, the abstract Shannon lower bound in \thmref{thm:slb} attains its largest value, the rate-distortion function, with the choice of $\mu$ in \eqref{eq:muopt}. The same choice of $\mu$ in \eqref{eq:C} and \eqref{eq:Ca} leads to
 \begin{align}
M &\geq \frac {\beta_{1 - \epsilon}(P_X, P_X \exp(\jmath_X(X, d)))}{\sup_{y \in \mathcal Y} \E{\exp(\jmath_X(X, d)) \1{\sd(X, y) \leq d}} } \label{eq:Cj},\\
\epsilon &\geq \sup_{\gamma > 0} \left\{ \Prob{ \jmath_X(X, d) \geq \log M + \gamma} - \exp(-\gamma)\right\} \label{eq:Caj}.
\end{align}
The bound in \eqref{eq:Caj} is just \cite[Theorem 7]{kostina2011fixed}. This bound is first- and second-order optimal, that is, for memoryless sources and separable distortion measures the converse part of the result in \eqref{block2order} can be recovered using \eqref{eq:Caj}. The bound in \eqref{eq:Cj}, which is new, is always better than \eqref{eq:Caj}, as the proof of \thmref{thm:C} shows. In the special cases of a binary source with Hamming distortion and the Gaussian source with mean-square distortion,  \eqref{eq:Cj} reduces to the bounds in \cite[Theorem 20]{kostina2011fixed} and \cite[Theorem 36]{kostina2011fixed}. A bound that is numerically tighter than \cite[Theorem 20]{kostina2011fixed} and \cite[Theorem 36]{kostina2011fixed} in some cases was recently proposed by Palzer and Timo \cite{palzer2016converse}. The bound in \cite{palzer2016converse} involves an optimization over an auxiliary scalar, while \eqref{eq:Cj} provides the tightest known general converse bound to date that does not require an optimization over auxiliary parameters. 

\detail{
Palzer and Timo \cite{palzer2016converse} propose
\begin{equation}
d\mu(x)  =  \1{\jmath_X(x, d) \geq \zeta} dP_X(x)
\end{equation}
}

\subsection{The necessary and sufficient condition}
\label{sec:condition}
The following result pins down the necessary and sufficient conditions for equality in \eqref{eq:slbgen} to hold. 
\begin{thm}
Assume that the infimum in \eqref{eq:RR(d)} is achieved by some $P_{Y^\star|X}$. Then, 
the following statements are equivalent. 
\begin{enumerate}[A.]
 \item The rate-distortion function is equal to Shannon's lower bound,
\begin{equation}
\mathbb R_X(d) =  h_\mu(x) - \phi_\mu(d). \label{eq:lbeq}
\end{equation}
\label{item:A}
\item For $P_X$-a.s. $x$, 
\begin{equation}
\jmath_X(x, d) =  \imath_\mu (x) - \phi_\mu(d). \label{eq:jlbeq}
\end{equation}
\label{item:B}
\item The backward conditional distribution\footnote{That is, $P_{X|Y^\star}$ such that $P_X P_{Y^\star|X} = P_{X | Y^\star} P_X$.} that achieves $\mathbb R_X(d)$  satisfies, for $P_{Y^\star}$-a.s. $y$, 
\begin{equation}
 \frac{d P_{X|Y^\star = y}}{d \mu}(x) = \frac{\exp(- \lambda \sd(x, y))}{\Sigma}.
\end{equation}
\label{item:back}

\end{enumerate}
\label{thm:lbeq}
\end{thm}

\begin{proof} 

\ref{item:B} $\Rightarrow$ \ref{item:A} is trivial. 
To show \ref{item:A} $\Rightarrow$ \ref{item:B}, note that the existence of $P_{Y^\star|X}$ that achieves the infimum in \eqref{eq:RR(d)} implies differentiability of $\mathbb R_X(d)$ \cite{csiszar1974extremum}. It follows that the maximum in \eqref{eq_RR(d)csiszar} is attained by a unique $g(x)$ \cite{csiszar1974extremum}. Since \eqref{eq:lbeq} establishes that $g(x)$ that attains the maximum in \eqref{eq_RR(d)csiszar} is that in Table \ref{tab:slb}, \eqref{eq:jlbeq} is immediate.

To show \ref{item:B} $\Leftrightarrow$ \ref{item:back}, recall the following equivalent representation of $\jmath_X(x, d)$ \cite{kostina2011fixed}:
\begin{equation}
 \jmath_X(x, d) = \log \frac{dP_{X|Y^\star = y}}{dP_{X}}(x) + \lambda \sd(x, y ) - \lambda d \label{eq:jdensity}.
\end{equation}
Equality in \eqref{eq:jdensity} holds for $P_{Y^\star}$-a.s. $y$. Comparing \eqref{eq:jlbeq} and \eqref{eq:jdensity} we conclude the equivalence
\ref{item:B} $\Leftrightarrow$ \ref{item:back}. 

\end{proof}

The necessary and sufficient conditions in Theorem \ref{thm:lbeq} assume a particularly simple form for difference distortion measures and the choice of $\mu$ as in Table \ref{tab:slbdiff}. In that case, clause \ref{item:back} can be replaced by
\begin{enumerate}[A.]
\item[C$^\prime$.]  There exists a random variable $Y^\star$ such that 
\begin{align}
X = Y^\star + Z_\lambda,  \label{eq:yza}
\end{align}
where $Y^\star$ is independent of $Z_\lambda$, and $Z_\lambda$ is defined in Table \ref{tab:slbdiff}.
\end{enumerate}

\begin{example}
 If $X$ is equiprobable on a finite group, \eqref{eq:lbeq} always holds. Indeed, in that case, equiprobable $Y^\star$ satisfies \eqref{eq:yza}. 
\end{example}

\begin{example}
For a binary $X$ with bias $p$ under Hamming distortion measure, $Z_\lambda$ is binary with bias $d$, and \eqref{eq:lbeq} is satisfied for all $0 \leq d \leq p$ by $Y^\star$ with bias $q$, where $q (1 - d)+  (1 - q) d = p$. 
\end{example}

\begin{example}
 Gaussian source with mean-square error distortion satisfies the conditions of Theorem \ref{thm:lbeq}; indeed, $X = Y^\star + Z_\lambda$, where $X \sim \mathcal N(0, \sigma^2\, \mathbf I )$, $Y^\star \sim  \mathcal N(0, (\sigma^2 - d) \, \mathbf I)\pperp  Z_\lambda \sim \mathcal N(0, d\, \mathbf I)$.
\end{example}

Theorem \ref{thm:lbeq} extends a result by Gerrish and Schultheiss \cite{gerrish1964nongaussian}, who showed that for the compression of a continuous random vector under the mean-square error distortion, the Shannon lower bound 
 gives the actual value of rate-distortion function if and only if $X$ can be written as the
sum of two independent random vectors $X= Y^\star+Z$,  where $Z \sim \mathcal N(0, d \, \mathbf I )$. Theorem \ref{thm:lbeq} also generalizes the backward-channel
condition for equality in the Shannon lower bound given in \cite[Theorem 4.3.1]{berger1971rate}. Unlike these classical results, Theorem \ref{thm:lbeq} applies to abstract sources and does not enforce any symmetry assumptions on the distortion measure.  

\todo{Corollary for stationary ergodic sources. Even for sources with memory, due to separability of distortion measure, $X_i = Y_i + Z_\lambda$ is necessary and sufficient. }

Most continuous probability distributions do not meet the condition in \eqref{eq:yza}.  In particular, an $X$ with indecomposable distribution cannot satisfy \eqref{eq:yza}, for any difference distortion measure.  In contrast, as the following result shows, for finite alphabet sources the classical Shannon lower bound (i.e. taking $\mu$ be the counting measure) is always attained with equality, as long as target distortion $d$ is not too large.  

\begin{thm}[Pinkston \cite{pinkston1969application}]
Let $X \in \mathcal X$, where $\mathcal X$ is a finite alphabet.  Let the distortion measure $\mathsf d \colon \mathcal X \mapsto \mathcal X$ be such that $\sd(x, x) = 0$,  $\sd(x, y) > 0$ for all $x \neq y$, and all columns of the 
distortion matrix $\{\sd(x, y)\}_{x, y}$ consist of the same set of entries (balanced distortion measure). 
Then, 
there exists a $d_c > 0$ such that the classical Shannon lower bound is satisfied with equality for  
\begin{equation}
 0 \leq \forall d \leq d_c. \label{eq:d0}
\end{equation}
\label{thm:pinkston}
\end{thm}

\begin{example}
For symbol error distortion equality in \eqref{eq:rfinite} holds for all
\begin{equation}
 0 \leq d \leq (m-1) \min_{x \in \mathcal X} P_{X}(x).
\end{equation}
\end{example}
Difference distortion measures satisfying \eqref{eq:d00} are included in the assumption of \thmref{thm:pinkston}. 
Generalizations of Pinkston's result are found in the works of Gray \cite{gray1971stationary,gray1970information,gray1971markov}, who showed in \cite{gray1971stationary} that the rate-distortion function equals the Shannon lower bound in the range of small distortions for stationary ergodic finite alphabet sources, generalizing and simplifying the proofs of Gray's previous results in \cite{gray1970information} (binary Markov source with BER distortion and Gauss-Markov source) and \cite{gray1971markov} (finite state finite alphabet Markov sources).  

Leveraging the necessary and sufficient conditions in Theorem \ref{thm:lbeq}, we conclude that under the conditions of Theorem \ref{thm:pinkston},  the $\sd$-tilted information is given by \eqref{eq:jfinite}. Applying \eqref{eq:jfinite} to \eqref{block2order}, we conclude that for the compression of a discrete memoryless sources under a balanced distortion measure, the minimum achievable rate compatible with excess probability $\epsilon$ at distortion $d$ satisfies \eqref{eq:Rapproxf}
for all $d \leq d_c$.

As mentioned above, continuous sources rarely meet the classical Shannon lower bound with equality, and thus \eqref{eq:Rapproxf} does not hold in general.
 Nevertheless, as we will see, lattice quantization of continuous sources often approaches \eqref{eq:Rapproxf}. This striking phenomenon is the major focus of the remainder of the paper.  The next section introduces the topic by discussing lattice coverings of space.

\begin{apxonly}
 By the Lebesgue decomposition theorem, a probability distribution can be uniquely represented by a mixture ($p + q + r = 1$) 
\begin{equation}
 \mu = p \mu_d + q \mu_c + r \mu_s
\end{equation}
where $\mu_d$ is purely atomic, $\mu_c$ is continuous wrt to Lebesgue, and $\mu_s$ is singular with respect to Lebesgue. 
TODO: a version of Shannon's lower bound for mixed distributions? 
\end{apxonly}

\section{Nonasymptotic analysis of lattice quantization}
\label{sec:lattice}
\subsection{Lattices: definitions}
Let $\mathcal C$ be a non-degenerate lattice in $\mathbb R^n$: 
\begin{equation}
\mathcal C \triangleq \{ c = \mathsf G\,  i \colon i \in \mathbb Z^n \} \label{eq:Gi},
\end{equation}
where the $n \times n$ generator matrix $\mathsf G$ is non-singular. 

\apxonl{
Let $\mathcal V_{\mathcal C}$ be the fundamental region of lattice $\mathcal C$, i.e.
\begin{equation}
\mathcal V_{\mathcal C} \triangleq \left\{ x \in \mathbb R^n \colon \| x\| \leq \| x - c \| ~ \forall c \in \mathcal C \right\}.   
\end{equation}
}
The nearest-neighbor $\mathcal C$-quantizer is the mapping $q_{\mathcal C} \colon \mathbb R^n \mapsto \mathcal C$ defined by
\begin{equation}
q_{\mathcal C} (x) \triangleq \argmin_{c \in \mathcal C} \|x - c\| \label{eq:q},
\end{equation}
and the Voronoi cell $\mathcal V_{\mathcal C}(c)$ is the set of all points quantized to $c$:
\begin{equation}
\mathcal V_{\mathcal C}(c) \triangleq \left\{ x \in \mathbb R^n \colon q_{\mathcal C} (x) = c \right\} \label{eq:vcell}.
\end{equation}
The ties in \eqref{eq:q} are resolved so that the resulting Voronoi cells are congruent. 
\apxonl{The Voronoi cell $\mathcal V_{\mathcal C}(0)$ associated with the origin ($c = 0$) 
is called the {\it fundamental} cell of lattice $\mathcal C$.}
We denote by $V_{\mathcal C}$ the cell volume of lattice $\mathcal C$:
\begin{align}
 V_{\mathcal C} &\triangleq \mathrm{Vol} \left( \mathcal V_{\mathcal C}(0) \right)\\
 &= |\det \mat G|  \label{eq:G}. 
\end{align}

The radius of the Voronoi cell, i.e. the minimum radius of a ball containing $\cV_\cC(0)$, is called the {\it covering radius} of lattice $\mathcal C$:
\begin{equation}
r_{\cC} \triangleq \max_{x \in \mathbb R^n} \|x - q_{\mathcal C}(x)\| \label{eq:rcov}.
\end{equation}
{\it Covering efficiency} of a lattice $\cC$ is measured by the normalized ratio of the volume of the ball of radius $r_\cC$ to the volume of the Voronoi cell: 
\begin{align}
\rho_{\cC} &\triangleq 
 r_\cC \frac{b_n^{\frac 1 n}}{V_{\cC}^{\frac 1 n}}, \label{eq:rhoMSE}
\end{align}
where $b_n$ is the volume of a unit ball:
\begin{equation}
 b_n \triangleq \frac{\pi^{\frac n 2}}{\Gamma \left(\frac n 2 + 1\right) } \label{eq:bn}.
\end{equation}
By definition, 
\begin{equation}
 \rho_{\cC} \geq 1,
\end{equation}
and the closer $\rho_{\cC}$ is to $1$ the more sphere-like the Voronoi cells of $\cC$ are and the better lattice $\cC$ is for covering.

\begin{apxonly}
\emph{This does not generalize to non-Euclidean distortion measures:}  
The effective radius of lattice $\mathcal C$ is the radius of a ball having the same volume as the cell volume:
\begin{equation}
r_\cC^{\mathrm{eff}} \triangleq \left[ \frac {V_\cC}{b_n}\right]^{\frac 1 n}   \label{-eq:reff}
\end{equation}
where $b_n$ is the volume of a unit ball:
\begin{equation}
 b_n \triangleq \frac{\pi^{\frac n 2}}{\Gamma \left(\frac n 2 + 1\right) } 
\end{equation}

By definition, 
\begin{equation}
 r_\cC^{\mathrm{eff}} \leq r_{\cC}, \label{-eq:efco}
\end{equation} 

{\it Covering efficiency} of lattice $\cC$ is measured by the ratio 
\begin{equation}
\rho_{\cC} \triangleq \frac{r_{\cC}}{r_{\cC}^{\mathrm{eff}} } \label{-eq:rho}
\end{equation}
\end{apxonly}

\subsection{The distribution at the output of lattice quantizers}

The purpose of this subsection is to express the distribution at the output of lattice quantizers in terms of the distribution of the raw data $X$ and the sizes of quantization cells. 
We will characterize both the {\it information} at the output of the quantizer, that is, the random variable
\begin{equation}
\imath(q_\cC(X)) = \log \frac 1 {q_\cC(X)} \label{eq:qinfo},
\end{equation}
and the corresponding entropy, $H(q_\cC(X))$, given by the expectation of \eqref{eq:qinfo}. This characterization holds the key to studying the fundamental limits of data compression with lattices. 
Indeed, it is well known that $L^\star_S$, the minimum average length required to losslessly encode discrete random variable $S$ is bounded in terms of the entropy of $S$ as \cite{alon1994lower,wyner1972upper}
\begin{align}\label{eq:alon2}
		H(S) - \log_2 (H(S)+1) - \log_2 e &\le L^\star_S \\
		&\le H(S).\, \label{eq:wyner}
\end{align}
Likewise, in fixed-length almost lossless data compression, the probability of error $\epsilon$ is bounded in terms of the distribution of the information random variable as (e.g. \cite{verdu2009notes})
\begin{align}
\Prob{\imath(S) \geq \log M + \gamma } - \exp(-\gamma) &\leq
\epsilon  \label{eq:infoc}\\
&\leq \Prob{ \imath(S) > \log M } \label{eq:infoa},
\end{align}
where $M$ is the number of distinct values at the output of the compressor. 

For a lattice $\mathcal C \in \mathbb R^n$, denote the  $n$-dimensional random vector $X_{\mathcal C}$ by
\begin{equation}
X_{\cC} \triangleq  q_\cC(X) + U_\cC  \label{eq:Xlambda},
\end{equation}
where the random vector $U_\cC$ is uniform on $\cV_\cC(0)$. See Fig. \ref{fig:latticebins}. 
\begin{figure}[h]
\centering
\includegraphics[width=.7\linewidth]{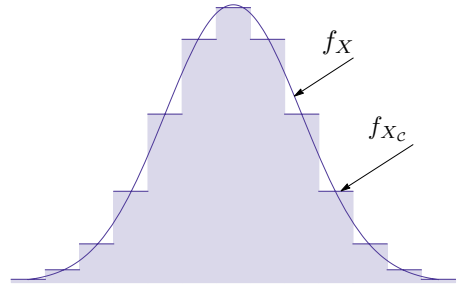}
\caption{Example: densities of $X$ and $X_\cC$ for $n = 1$. }
\label{fig:latticebins}
\end{figure}

It is easy to verify that the entropy at the output of the quantizer based on lattice $\cC$ can be expressed as

\begin{equation}
H\left( q_\cC(X) \right) = h(X_\cC) - \log {V_\cC}.  
\label{eq:renyig}
\end{equation}
If $h(X) > -\infty$, then
\begin{align}
D(X \| X_\cC) 
&= - h(X) +  h(X_\cC) \label{eq:ver},
\end{align}
and \eqref{eq:renyig} can be rewritten as
\begin{align}
H\left( q_\cC(X) \right) &= h(X) - \log V_\cC + D(X \| X_\cC) \label{eq:renyig1}.
\end{align}

If the distortion measure is the mean-square error, the maximum distortion is related to the covering radius as
\begin{equation}
 d = \frac {r_\cC^2} n, \label{eq:dr}
\end{equation}
so that
\begin{align}
 \log V_\cC 
 &= n \log  \sqrt {nd} + \log b_n - n \log \rho_{\cC}, \label{eq:Vd}
\end{align} 
and we can continue \eqref{eq:renyig1} as
\begin{align}
H\left( q_\cC(X) \right)
&= h(X) - n \log \sqrt{nd} - \log b_n 
\notag\\ &
+ n \log \rho_{\cC} +  D(X \| X_\cC).
\label{eq:Hqd} 
\end{align}
By Stirling's approximation, 
\begin{align}
\log b_n &= \frac{n} 2 \log \frac {2 \pi e} n - \frac 1 {2 } \log n + \bigo{1  }. \label{eq:Ana}
\end{align}
It follows that the difference between the first three terms in \eqref{eq:Hqd} and Shannon's lower bound is of order  $\frac 1 {2 } \log n + \bigo{1  }$, as $n \to \infty$. The term $D(X \| X_\cC)$ is the penalty due to $f_X$ not being completely flat. Intuitively, this term decreases as the quantization cells shrink, an effect we will explore in detail shortly (see \thmref{thm:entropyq} and \thmref{thm:hq} below). The term $n \log \rho_{\cC}$ is the penalty due to the lattice cells not being perfect spheres. To understand how large this term can be, note that the thinnest lattice covering in dimensions 1 to 5 is proven to be $A_n^*$  (Voronoi's principal lattice of the first type) \cite{conway2013sphere}, which has covering efficiency
\begin{equation}
 \rho_{A_n^*} = b_n^{\frac 1 n} (n + 1)^{\frac 1 {2n}} \sqrt{ \frac {n(n + 2)}{12 (n + 1)} },
\end{equation}
so for $n = 1, 2, \ldots, 5$, 
\begin{align}
H\left( q_{A_n^*}(X) \right)
&= h(X) - n \log \sqrt{nd}  
+  \frac n 2 \log \frac {n(n + 2)}{12 (n + 1)} 
\notag\\ &
 + \frac 1 2 \log (n + 1) +  D(X \| X_\cC).
\label{eq:HqAn} 
\end{align}
Actually, $A_n^*$ is the thinnest lattice covering known in all dimensions $n \leq 23$. But $A_{24}^*$ has covering efficiency $\approx 1.189$ and is inferior to the Leech lattice $\Lambda_{24}$, for which $\rho_{\Lambda_{24}} \approx 1.090$. 

\apxonl{
\begin{align}
\mathbb L_{X^n}(d)
&=  h(X^n) + \frac{n}{2} \log \frac 1 {d} 
+  \frac n 2 \log \frac {n(n + 2)}{12 (n + 1)}  + \frac 1 2 \log (n + 1)  + \smallo{1}. \label{}
\end{align}

Given a point in $\mathbb R^n$, the nearest lattice point of $A_n^\star$ can be computed in $\bigo{n \log n}$ steps \cite{vaughan1999algorithm}. 
}
More generally, the following result demonstrates the existence of covering-efficient lattices.

\begin{thm}[{Rogers \cite[Theorem 5.9]{rogers1964packing}}]
For each $n \geq 3$, there exists an $n$-dimensional lattice $\cC_n$  with covering efficiency
\begin{equation}
n \log \rho_{\cC_n} \leq \log_2 \sqrt{2 \pi e} \left( \log n + \log \log n + c\right), \label{eq:rogers}
\end{equation}
where $c$ is a constant.
\label{thm:rogers}
\end{thm}

A natural question to ask next is the following: what is the minimum $H\left( q_\cC(X) \right)$ attainable among all lattice quantizers? For a distortion measure $\mathsf d \colon \mathbb R^n \times \mathbb R^n \mapsto \mathbb R^+$, denote the minimum entropy at the output of a lattice quantizer for the random vector $X \in \mathbb R^n$ by
\begin{equation}
\mathbb L_{X}(d) \triangleq  \inf_{\mathcal C \colon \mathsf d(X, q_\cC(X)) \leq d \text{ a.s.}} H( q_\cC(X)) \label{eq:Ld}.
\end{equation}
The definition in \eqref{eq:Ld} parallels the definition of $d$-entropy \cite{posner1967epsilonentropy}: \footnote{In literature, the distortion threshold is sometimes denoted by $\epsilon$ and the quantity in \eqref{eq:Hd} is referred to as the ``epsilon-entropy''. In this paper, $\epsilon$ is reserved for the excess distortion probability and $d$ for the distortion threshold.}
\begin{equation}
\mathbb H_{X}(d) \triangleq  \inf_{\substack{ q \colon \mathbb R^n \mapsto \mathbb R^n \\ \mathsf d(X, q(X)) \leq d \text{ a.s.}}} H( q(X)) \label{eq:Hd},
\end{equation}
with the distinction that in \eqref{eq:Hd} the infimization is performed over all mappings $q \colon \mathbb R^n \mapsto \mathbb R^n$ and not just lattice quantizers. For that reason, we call the function in \eqref{eq:Ld} {\it lattice $d$-entropy}. Note that 
\begin{equation}
 \mathbb R_{X}(d) \leq \mathbb H_{X}(d) \leq \mathbb L_{X}(d). \label{eq:RH}
\end{equation}

Using \eqref{eq:Hqd}, if $h(X) > -\infty$, we can characterize the lattice $d$-entropy under the mean-square error distortion \eqref{eq:mse} as, 
\begin{align}
\mathbb L_{X}(d)
&=  h(X) - n \log \sqrt{nd} - \log b_n 
\notag\\
&
+ \inf_{\mathcal C \colon r_\cC \leq \sqrt{nd}}  \left\{D(X \| X_\cC) + n \log \rho_{\cC} \right\}.
\label{eq:Ldchar}
\end{align}
Since the term inside the infimum in \eqref{eq:Ldchar} is nonnegative, $\mathbb L_{X}(d)$ is lower-bounded by the first three terms in \eqref{eq:Ldchar}. Applying \eqref{eq:Ana} we see that these three terms are within $\frac 1 {2 } \log n + \bigo{1  }$  information units from Shannon's lower bound, so 
\begin{equation}
 \mathbb L_{X}(d) \geq h(X) - n  \log  \sqrt {2 \pi e d} + \frac 1 {2 } \log n + \bigo{1  } \label{eq:Lnl}. 
\end{equation}
On the other hand, upper-bounding the infimum in \eqref{eq:Ldchar} by picking any lattice $\mathcal C$ that satisfies the Rogers condition \eqref{eq:rogers}, we obtain
\begin{equation}
 \mathbb L_{X}(d) \leq h(X) - n  \log  \sqrt {2 \pi e d}  + D(X \| X_\cC) + \bigo{\log n  } \label{eq:Lnu}. 
\end{equation}

 As we will see shortly, for small $d$, the term $D(X \| X_\cC)$ is also small. Intuitively, this means that at large $n$ and small $d$, good lattice quantizers are almost as good as the best optimal quantizer.

In the rest of this section, we explore the behavior of $h(X_\cC)$ and $D(X \| X_\cC)$. The next result, \thmref{thm:entropyq} below, concerns itself with the behavior of $h(X_\cC)$ in the limit of increasing point density, or vanishing cell volume.  As evident from \eqref{eq:G}, a scaling of $\mat G$ by $\frac{V^{\frac 1 n}}{|\det \mat G|^{\frac 1 n} }$ results in the lattice of cell volume $V$. Fixing $\mat G$ and considering lattices generated by $\frac{V^{\frac 1 n}}{|\det \mat G|^{\frac 1 n} } \mat G$, we obtain a continuum of lattices parameterized by $V$. We will be interested in the behavior of $h(X_\cC)$ as $V \to 0$. 
Clearly, as quantization cells become smaller, the distribution of $X_\cC$ becomes a better approximation of the distribution of $X$. \thmref{thm:entropyq} below, due to Csisz\'ar \cite{csiszar1971entropyquantization},  formalizes this intuition by underlining the connection between the entropy of $X_\cC$ and the differential entropy of $X$.

For vector $x^n \in \mathbb R^n$, $\lfloor x^n \rfloor$ denotes the vector of integer parts of its components, that is, $\lfloor x^n \rfloor = (\lfloor x_1 \rfloor, \ldots,  \lfloor x_n \rfloor)$.

\begin{thm}[Csisz\'ar \cite{csiszar1971entropyquantization}]
Let $X$ be a random variable and let $\mathcal C$ be a lattice in $\mathbb R^n$. Assume that 
\begin{equation}
H(\lfloor X \rfloor) < \infty. \label{eq:entropy}
\end{equation}
For any sequence of lattices with vanishing cell volume,
\begin{equation}
 \lim_{V_{\cC} \to 0} h(X_\cC) = h(X),  \label{eq:hXC}
\end{equation}
where $X_\cC$ is defined in \eqref{eq:Xlambda}.
\label{thm:entropyq}
\end{thm}

\detail{
\begin{proof}
Note that $X_\cC$ has a density even if $X$ does not. The distribution of discrete random variable $q_\cC(X)$ satisfies, for any $c \in \mathcal C$,
\begin{align}
 \frac 1 {V_\cC} P_{q_\cC(X)}(c) &=  f_{X_\cC}( c - u_0), \qquad \forall u_0 \in \cV_\cC(0) \label{eq:q1}
 \end{align}
Taking logarithms on both sides of \eqref{eq:q1} and then an expectation with respect to $f_{X_\cC}$  reveals \eqref{eq:renyig}. 

To show \eqref{eq:hXC}, continue \eqref{eq:q1} as 
\begin{align}
f_{X_\cC}( c - u_0)  &= \E{ f_X(c - U_{\mathcal C}) }  \label{eq:q2} 
\end{align}
By Jensen's inequality
\begin{align}
 h(X) \leq h(X_\cC) \label{eq:jf}
\end{align}
If $h(X) = + \infty$, due to \eqref{eq:jf} there is nothing to prove.  For $h(X) < +\infty$, the validity of \eqref{eq:hXC} under assumption \eqref{eq:entropy} is shown in \cite{csiszar1971entropyquantization}.

\apxonl{If the sequence of lattices in \eqref{eq:hXC} is a sequence of nested lattices forming an increasing sequence of partitions, then using the same reasoning as in \eqref{eq:q2} and \eqref{eq:jf} and applying the monotone convergence theorem, we conclude that $h(X_\cC)$ monotonically decreases to $ h(X)$. In the more general case of non-nested lattices, the validity of \eqref{eq:hXC} under assumption \eqref{eq:entropy} is shown in \cite{csiszar1971entropyquantization}. } 

\apxonl{
CALCULATIONS:
\apxonl{
\begin{align}
H(q_\cC(X)) &= \sum_{c \in \cC} P_{q_{\cC}(X) } (c)  \log \frac 1  {P_{q_{\cC}(X) } (c)}\\
&= \sum_{c \in \cC} V_\cC f_{X_\cC}( c - u)   \log \frac 1  {V_\cC f_{X_\cC}( c - u)} \\
&= \sum_{c \in \cC} \int_{\mathcal V_\cC(0)} f_{X_\cC}( c - u)  \log \frac 1  {V_\cC f_{X_\cC}( c - u)} du\\
&= h(X_\cC) + \log \frac 1 {V_\cC} \label{eq:vera}
\end{align}
}
Using \eqref{eq:q2}, it is easy to verify that
\begin{align}
D(X \| X_\cC) 
&= - h(X) +  h(X_\cC) \label{eq:ver}
\end{align}
\apxonl{
\begin{align}
D(X \| X_\cC) &= 
 - h(X) - \sum_{c \in \cC} \int_{\mathcal V_\cC(0) } f_{X}(c - u) \log f_{X_\cC}(c - u) du
\\
&= 
 - h(X) - \sum_{c \in \cC} \log f_{X_\cC}(c - u) \int_{\mathcal V_\cC(0) } f_{X}(c - u)  du
\\
&= - h(X) -  \sum_{c \in \cC} \int_{\mathcal V_\cC(0) } f_{X_\cC}(c - u) \log f_{X_\cC}(c - u) du\\
&= - h(X) +  h(X_\cC) \label{eq:ver}\\
&= -h(X) + H(q_\cC(X)) - \log \frac 1 {V_\cC}
\end{align}
}
Combining \eqref{eq:vera} and \eqref{eq:ver}, we obtain \eqref{eq:renyi}. 
}

\end{proof}
}

\begin{apxonly}
Theorem \ref{thm:entropyq} is a generalization of the following: 
 in 
\cite[Theorem 2.30]{verduIT,renyi1959dimension}, it is shown that if $X_\Delta$ is a  $\Delta$-histogram of $X$, i.e.
\begin{equation}
X_\Delta = \left( \left \lfloor \frac X \Delta \right \rfloor + U \right) \Delta
\end{equation}
where $U$ is uniform on $(0,1)$, then
\begin{equation}
H\left( \left \lfloor \frac{X} \Delta \right \rfloor \right) = h(X) + \log \frac 1 {\Delta} + D(X \| X_\Delta)
\end{equation}
Moreover, as Renyi showed, 
\begin{equation}
 \lim_{\Delta \to 0} D(X \| X_\Delta) = 0. 
\end{equation}
\end{apxonly}

Theorem \ref{thm:entropyq} holds even if $X$ does not have a density; in that case, $h(X) = -\infty$. If $h(X) > -\infty$, using \eqref{eq:ver}, we can rewrite \eqref{eq:hXC} as
\begin{equation}
\lim_{V_{\cC} \to 0} D(X \| X_{\cC}) = 0. \label{eq:kl}
\end{equation}

 Theorem \ref{thm:entropyq} also holds for the more general case of non-lattice partitions of $\mathbb R^n$ into sets of equal volume.     

Assumption \eqref{eq:entropy} is needed to ensure that the tails of $f_X(x) \log f_X(x)$ are well behaved. If the probability density function $f_X$ is continuous and is supported on a compact set, then one can show that \eqref{eq:hXC} holds in the following elementary manner (cf. \cite[Sec. 8.3]{cover2012elements}).
Applying the mean value theorem to 
\begin{align}
f_{X_\cC}( c - u_0)  &= \E{ f_X(c - U_{\mathcal C}) }  \label{eq:q2}, 
\end{align}
 for each $c \in \mathcal C$ we note the existence of $u_c \in \cV_\cC(0)$ such that 
\begin{equation}
 f_X(c - u_c) = f_{X_c}(c - u_0), \qquad \forall u_0 \in \cV_\cC(0). \label{eq:uc}
\end{equation}
It follows that $h(X_\cC)$ is the Riemann sum for $f_X(x) \log \frac 1 { f_X(x)}$ and the partition generated by the Voronoi cells of $\mathcal C$ labeled by $c - u_c, c \in \mathcal C$, that is, 
\begin{align}
h(X_\cC) &=  \sum_{c \in \cC} V_\cC f_{X_\cC}(c) \log \frac 1 {f_{X_\cC}(c) }\\
&= \sum_{c \in \cC} V_\cC f_X(c - u_c) \log \frac 1 {f_{X}(c - u_c) }.
\end{align}
 Convergence to $h(X)$ follows by the definition of the Riemann integral.

  The sufficient condition for \eqref{eq:entropy} to hold is 
 \begin{align}
 \E{\log \left(1 + \frac 1 {\sqrt n }\|X\| \right)} < \infty, \label{eq:wu}
\end{align}
which in turn holds for any source vector with $\E{\|X\|^\alpha} < \infty$ for some $\alpha > 0$. This sufficient condition was proved by Wu and Verd\'u \cite[Proposition 1]{wu2010renyi} for scalar $X$; Koch \cite{koch2015shannonlb} noticed that it continues to hold in the vector case as well.

Prior to Csisz\'ar, the validity of \eqref{eq:hXC} under a more restrictive assumption was proved by R\'enyi \cite[Theorem 4]{renyi1959dimension}. 
 Csisz\'ar \cite{csiszar1971entropyquantization} showed the validity of \eqref{eq:hXC} under the following assumption.
 Suppose there exists some Borel measurable partition $\{\mathcal B_1, \mathcal B_2, \ldots\}$ of $\mathbb R^n$ into sets of finite Lebesgue measure such that the following two conditions are satisfied.
\begin{enumerate}
\item 
 \begin{equation}
\sum_{i} P_X(\mathcal B_i) \log \frac 1 {P_X(\mathcal B_i)} < \infty.
\end{equation}
\item 
There exist $\rho > 0$ and $s \in \mathbb N$ such that for all $k$, the distance between $\mathcal B_k$ and $\mathcal B_\ell$, $k \neq \ell$, is greater than $\rho$ for all but at most $s$ indices $\ell$. 
\end{enumerate}
Recently, Koch \cite{koch2015shannonlb} showed that the rate-distortion function is infinite for all $d>0$ if $H(\lfloor X \rfloor) = \infty$, as long as the difference distortion measure has form $\|\cdot\|^r$, where $\| \cdot \|$ is an arbitrary
norm and $r > 0$. Thus, the assumption \eqref{eq:entropy} is as general as Csisz\'ar's assumption in most cases of interest.

\apxonl{
\begin{thm}
 \begin{align}
\Var{\log {f_{X_\cC}(X_\cC) }} \leq  \Var{\log f_X(X)} + \frac{\log^2 e}{e} \E{ \left( 2 f_X(X) - e\right) 1\{ f_{X}(X) > e\} } \label{eq:vardiv}
\end{align}
\end{thm}

\begin{proof}
We use \eqref{eq:jf} to write
\begin{equation}
 \Var{\log {f_{X_\cC}(X_\cC) }} \leq  \E{\log^2 {f_{X_\cC}(X_\cC) }} -  h^2(X).
\end{equation} 
Observe that the function $\log^2 z$ is convex for $z < e$ and concave elsewhere. Furthermore, it is bounded from above by the following convex function:
\begin{equation}
 \log^2 z \leq \log^2 z\, 1\{ z \leq e\} + \frac{\log^2 e} e \left(  2  z - e \right) 1\{ z > e\}
\end{equation}
By Jensen's inequality, 
\begin{align}
  \E{\log^2 {f_{X_\cC}(X_\cC) }} &\leq \E{\log^2 {f_{X}(X) }} + \frac{\log^2 e}{e} \E{ \left( 2 f_X(X) - e\right) 1\{ f_{X}(X) > e\} },
\end{align}
and \eqref{eq:vardiv} follows. 
\end{proof}
}

The strength of Theorem \ref{thm:entropyq} is that it requires only a very mild assumption on the source density, namely, \eqref{eq:entropy}. The weakness is that it does not offer any estimate on the speed of convergence to the limit in \eqref{eq:hXC} (or, equivalently, in \eqref{eq:kl}); such an estimate will be crucial in our study of the behavior of the output distribution of lattice quantizers in the limit of increasing dimension. Naturally, for the relative entropy $D(X \| X_\cC)$ to be small, the probability density function of $X$ should not change too abruptly within a single quantization cell.

The following smoothness condition will be instrumental in quantifying the variability of $f_X$.  

\begin{defn}[$v$-regular density]
Let $v \colon \mathbb R^n \mapsto \mathbb R^+$. Differentiable probability density function $f_X$ is called $v$-regular if
\begin{equation}
\| \nabla f_X(x)\| \leq v(x) f_X(x),  \qquad \forall x \in \mathbb R^n. 
\label{eq:pwreg}
\end{equation}
\label{defn:pwreg}
\end{defn}

All differentiable probability density functions are $v$-regular, for some $v$. 
The function $v(x)$ measures how fast the probability density function $f_X(x)$ varies as $x$ varies. In the extreme case of $X$ uniform on $\Omega$, an open subset of  $\mathbb R^n$, we have $v(x) \equiv 0$, $\forall x \in \Omega$. In general, the function $v(x) \geq 0$ can be thought of as the measure of the distance between $f_X$ and the uniform distribution. In fact, $v(x)$ is closely related to the \emph{total variation} of $f_X(x)$:
\begin{align}
\mathrm{TV}(f_X) &\triangleq \int_{\mathbb R^n}  \| \nabla f_X(x)\| d x \\
&\leq \E{v(X)}. \label{eq:tvv}
\end{align}
Using \eqref{eq:tvv}, we see that a differentiable $f_X$ has finite variation if and only if \eqref{eq:pwreg} holds with equality for $f_X$-a.s. $x$ and some function $v$ such that $\E{v(X)} < \infty$; the total variation of $f_X$ is then given by $\E{v(X)}$. 

Another way to look at \eqref{eq:pwreg} is to observe that at any $x$ with $f_X(x) > 0$, \eqref{eq:pwreg} is equivalent to 
\begin{equation}
 \| \nabla \log f_X(x)\| \leq v(x) \log e \label{eq:logpw}.
\end{equation}
Thus, $v(x)$ can be taken to be the norm of the gradient of the natural logarithm of  $f_X(x)$, which results in equality in~\eqref{eq:logpw} and is thereby optimal.
For example, if $X \sim \mathcal N(0, \sigma^2\, \mathbf I)$, then the optimal choice is $v(x) = \frac {2}{\sigma^2} \| x \|$. 

Since the function $v(x)$ quantifies how much the density of $X$ can change within a single quantization cell, it will be useful in bounding the entropy (and information) at the output of a lattice quantizer for $X$ in terms of the sizes of lattice quantization cells. 

Definition \ref{defn:pwreg} presents a generalization of a smoothness condition recently suggested by Polyanskiy and Wu \cite{polyanskiy2015wasserstein}, who considered densities satisfying \eqref{eq:logpw} with
\begin{equation}
 v(x) = c_1 \|x\| + c_0 \label{eq:vpw}.
\end{equation} 
for some $c_0 \geq 0$, $c_1 > 0$. 

 A wide class of $c_1 \|x\| + c_0$~-regular densities is identified in \cite{polyanskiy2015wasserstein}. In particular, the density of $B + Z$, with $B \pperp Z$ and $Z \sim \mathcal N(0, \sigma^2\, \mathbf I)$ is $\frac {4}{\sigma^2} \E{\|B\|} + \frac {2}{\sigma^2} \| x \|$~-regular. Likewise, if the density of $Z$ is $c_1 \|x\| + c_0$~-regular, then that of $B + Z$, where $\|B\| \leq b$ a.s., $B \pperp Z$, 
is  $c_1 \|x\| + c_0 + c_1 b$~-regular. Furthermore, if $X$ has $c_1 \|x\| + c_0$-regular density and finite second moment, then its differential entropy is finite.

Regularity of a product density is easily established if the marginal densities are regular, as the following result details. 
\begin{prop}
If $f_{X^n} = f_{X_1} \ldots f_{X_n}$ and $f_{X_i}$ is $v_i$-regular, then $f_{X^n}$ is $v$-regular, where
\begin{equation}
 v(x^n) =\| v_1(x_1), \ldots, v_n(x_n) \|. \label{eq:vn}
\end{equation}
\label{prop:regprod}
\end{prop}
\begin{proof}
Since $f_{X_i}$ is $v_i$-regular, we have: 
\begin{equation}
\left | f_{X_i}^\prime(x_i) \right | \leq v_i(x_i)\,  f_{X_i}(x_i).
\end{equation}
Therefore,
\begin{align}
 &~
 \left\| \nabla f_X(x)\right\| 
 \\
 =&~\left\| \left[ \frac{f_{X_1}^\prime(x_1)}{f_{X_1}(x_1)}, \ldots,  \frac{f_{X_n}^\prime(x_n)}{f_{X_n}(x_n)} \right] \right\| f_{X_1}(x_1)\cdot \ldots \cdot f_{X_n}(x_n) \notag\\
 \leq&~ v(x^n) f_{X^n}(x^n).
\end{align}
\end{proof}

We now state the main result of Section \ref{sec:lattice}, which provides upper bounds on $\imath (q_\cC(x))$ and $H\left( q_\cC(X) \right)$ for regular densities. 
\begin{thm}
Let $X$ be a random variable with $h(X) > -\infty$ and $v$-regular density. Let $\mathcal C$ be a lattice in $\mathbb R^n$. Then the information random variable and the entropy at the output of lattice quantizer can be bounded as
\begin{align}
\imath (q_\cC(x)) &\leq \log \frac 1 {f_{X}(x)} - \log V_\cC  
+
2 r_{\cC} v_\cC(x) \log e, \label{eq:iqv}\\
H\left( q_\cC(X) \right) &\leq h(X) - \log V_\cC  + 2 r_\cC   \E{v_\cC \left( X  \right) } \log e, 
\label{eq:hq}
\end{align}
where $r_\cC$ is the lattice covering radius defined in \eqref{eq:rcov}, and $v_{\cC}(x)$ is given by
\begin{equation}
v_{\cC}(x) = \max_{u \in \mathcal V_\cC(q_\cC(x))} v(u). \label{eq:vv}
\end{equation}
Furthermore, if $v(x) = v(\|x\|)$ is convex and nondecreasing, then \eqref{eq:iqv} and \eqref{eq:hq} can be strengthened by replacing \eqref{eq:vv} with
\begin{equation}
v_{\cC}(x) =   \frac 1 2 v(\|x\|) + \frac 1 2 v(\|x\| + 2 r_\cC) \label{eq:v2},
\end{equation}
and  if $v(x) = c_1 \|x\| + c_0$, then \eqref{eq:v2} particularizes as 
\begin{equation}
v_{\cC}(x) = c_1 \|x\| + c_1 r_\cC + c_0 \label{eq:vpw}.
\end{equation}
\label{thm:hq}
\end{thm}

\begin{proof}

Observe that, 
\begin{align}
&~ \left|  \log f_X(a) - \log f_X(b)\right| \notag \\
 =&~   \left| \int_0^1 \left( \nabla \log f_X(t a + (1 - t)b), a - b \right) dt \right|\\
 \leq&~  \|a - b \| \log e  \int_0^1 v( t a + (1 - t)b ) dt  \label{eq:praa}\\
 \leq&~ \max_{0 \leq t \leq 1}  v(t a + (1 - t)b) \| a - b \|  \log e \label{eq:pra},
\end{align}
where $(\cdot, \cdot)$ denotes the scalar product, and \eqref{eq:praa} holds by Cauchy-Schwarz inequality. 
Using \eqref{eq:uc} and \eqref{eq:pra}, we evaluate the information in $q_\cC(x)$ as
\begin{align}
&~
\imath (q_\cC(x)) 
\notag\\
=&~\log \frac 1 {f_{X_\cC}(x)} - \log V_\cC\\
=&~ \log \frac 1 {f_{X}(x)} - \log V_\cC + \log \frac {f_{X}(x)} {f_{X}(q_\cC(x) - u_c) } \label{eq:iq0}\\
\leq&~ \log \frac 1 {f_{X}(x)} - \log V_\cC  
+
2 r_{\cC} v_\cC(x) \log e,  \label{eq:iqv0}
\end{align}
which is equivalent to \eqref{eq:iqv}, and \eqref{eq:hq} is immediate upon taking an expectation of \eqref{eq:iqv}.


If $v(x) = v(\|x\|)$, convex and nondecreasing, we strengthen \eqref{eq:pra} by applying Jensen's inequality to \eqref{eq:praa}: 
 \begin{align}
 &~
 \left|  \log f_X(a) - \log f_X(b)\right| 
 \notag\\
 \leq&~ \frac {\log e} 2 (v(\|a\|) + v(\|b\|)) \| a - b \|  \label{eq:Duconv1} \\
 \leq&~ \log e (v(\|a\|) + v(\|a\| + 2 r_\cC)) r_\cC, \label{eq:tria}
\end{align}
where to get \eqref{eq:tria} we used the triangle inequality, the assumption that $v$ is nondecreasing, and the fact that $\|a - b\| \leq 2 r_\cC$ for $a$ and $b$ from the same quantization cell. 
Modifying \eqref{eq:iqv0} accordingly results in a strengthening of \eqref{eq:iqv} and \eqref{eq:hq} with $v_\cC$ in \eqref{eq:v2}.



\end{proof}

 Comparing \eqref{eq:renyig1} and \eqref{eq:hq}, we see that \thmref{thm:hq} establishes
\begin{equation}
 D(X \| X_{\cC}) \leq 2 r_\cC   \E{v_\cC \left( X  \right) } \log e \label{eq:divu}.
\end{equation}
 Note that $v(x) = v_\cC \left( x  \right) \equiv 0$ if and only if $X$ is uniform. In that case, $D(X \| X_\cC) = 0$, and the third term in \eqref{eq:hq} vanishes. Otherwise, the third term in \eqref{eq:hq} is positive. It becomes larger if $f_X$ varies noticeably within each quantization cell, and it vanishes as the sizes of the quantization cells become smaller ($r_\cC \to 0$). Thus, Theorem \ref{thm:hq} allows one to quantify the convergence rate in Csisz\'ar's \thmref{thm:entropyq}.   

Succeeding a lattice quantizer $\cC$ by an optimal lossless coder (see Fig. \ref{fig:latticelossless}) and keeping only $M$ most likely realizations of the output of the lossless coder, one obtains, according to \eqref{eq:infoa}, an $(M, d, \epsilon)$ code with
\begin{equation}
\epsilon  \leq \Prob{ \imath (q_\cC(X)) > \log M } \label{eq:A}.
\end{equation}

\begin{figure}[hb]
 \center{\includegraphics[width=.8\linewidth]{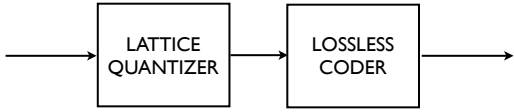}}
\caption{Separated architecture of lattice quantization.  
}
\label{fig:latticelossless}
\end{figure}

Applying the upper bound on $\imath(q_\cC(X))$ in \eqref{eq:iqv} to \eqref{eq:A}, we conclude that there exists an $(M, d, \epsilon)$ lattice code with
\begin{align}
 \epsilon &\leq \Prob{\log \frac 1 {f_{X}(X)}- \log V_\cC + 2 r_{\cC} v_\cC(X) \log e  > \log M }\\
 &\leq \Prob{\log \frac 1 {f_{X}(X)} - \log V_\cC + \gamma  > \log M } 
 \notag\\
&
 + \Prob{ 2 r_{\cC} v_\cC(X) > \gamma}, \label{eq:agamma}
\end{align}
where \eqref{eq:agamma} holds for any $\gamma \geq 0$ by the union bound. Applying \eqref{eq:Vd}, \eqref{eq:Ana} and \eqref{eq:rogers}  to \eqref{eq:agamma}, we obtain
\begin{align}
\epsilon &\leq \Prob{\ushort \jmath(X, d) + \gamma + \bigo{\log n} > \log M } 
\notag\\
&+ \Prob{ 2 r_{\cC} v_\cC(X) > \gamma}, \label{eq:agamma2}
\end{align}
where $\ushort \jmath(X, d)$ is defined in Table \ref{tab:slbdiff}. 
Furthermore, as we will show in \secref{sec:asymp}, under regularity conditions $\gamma = n \bigo{\sqrt d}$ can be chosen so that the second term in \eqref{eq:agamma2} is negligible, implying that
\begin{align}
\epsilon &\lessapprox  \Prob{\ushort \jmath(X, d) + n \bigo{\sqrt d} + \bigo{\log n}  > \log M }, \label{eq:Aapprox}
\end{align}
which provides a matching upper bound for \eqref{eq:Capprox}. Together, \eqref{eq:Capprox} and \eqref{eq:Aapprox} say that the excess distortion probability of the best code is given roughly by the complementary cdf of $\ushort \jmath(X, d) $ evaluated at $\log M$, the logarithm of the code size.  In other words, as advertised in \eqref{eq:jappx}, $\ushort \jmath(X, d) $ approximates the amount of information that needs to be stored about $X$ in order to restore it with distortion $d$.

\section{Asymptotic analysis of lattice quantization}
\label{sec:asymp}

\subsection{First order analysis}
\label{sec:asymp1}
Lattice coverings of space become more efficient as the dimension increases. In this section we study the fundamental rate-distortion tradeoffs attainable by lattice quantizers in the limit of large dimension $n$. This analysis is afforded by the bounds presented \secref{sec:slb} and \secref{sec:lattice}. \secref{sec:asymp1} presents the {\it first-order} (Shannon-type) asymptotic results comparing the behavior of lattice quantizers in the limit of infinite $n$ to the Shannon lower bound. A refined, {\it second-order} analysis quantifying how fast this asymptotic limit is approached is presented in \secref{sec:asymp2} below.

The rate-distortion function can be defined as follows. 
\begin{defn}
The rate-distortion function for the compression of a sequence of random variables $X_1, X_2, \ldots$ is defined by
\begin{equation}
R(d) \triangleq  \limsup_{n \to \infty} \frac 1 n \mathbb H_{X^n}(d) \label{eq:Rdlim}
\end{equation}
where $\mathbb H_{X^n}(d)$ is the $d$-entropy of vector $X^n$, defined in \eqref{eq:Hd}.
\end{defn}

The lattice rate-distortion function can be defined as follows. 
\begin{defn}
The lattice rate-distortion function for the compression of a sequence of random variables $X_1, X_2, \ldots$ is defined by
\begin{equation}
L(d) \triangleq  \limsup_{n \to \infty} \frac 1 n \mathbb L_{X^n}(d) \label{eq:Ldlim}
\end{equation}
\end{defn}

The operational meaning of \eqref{eq:Rdlim} is the minimum average rate asymptotically compatible with maximal distortion $d$. 
Indeed, substituting $S = q(X^n)$ into \eqref{eq:alon2} and \eqref{eq:wyner} and dividing through by $n$, we conclude that \eqref{eq:Rdlim} is equivalent to the operational definition
\begin{equation}
R(d) = \limsup_{n \to \infty} \frac 1 n \inf_{\substack{ q \colon \mathbb R^n \mapsto \mathbb R^n \\ \mathsf d(X^n, q(X^n)) \leq d \text{ a.s.}}} L^\star_{q(X^n)}. 
\end{equation}
Similarly, \eqref{eq:Ldlim} is equivalent to the operational definition
\begin{equation}
L(d) = \limsup_{n \to \infty} \frac 1 n \inf_{\mathcal C \colon \mathsf d(X^n, q_\cC(X^n)) \leq d \text{ a.s.}} L^\star_{q_\cC(X^n)}. 
\end{equation}

For convenience, denote the limsup of normalized $n$-dimensional Shannon's lower bounds 
\begin{equation}
 \ushort R(d) \triangleq \limsup_{n \to \infty} \frac 1 n \ushort{ \mathbb R}_{X^n}(d),
\end{equation}
where recall that $\ushort{ \mathbb R}_{X}(d)$ denotes the classical Shannon lower bound for $X$.

The first result in this section provides a characterization of the lattice rate-distortion function for sources with regular densities. 
\begin{thm}
 Consider a random process $X_1, X_2, \ldots$. The lattice rate-distortion function under the mean-square error distortion satisfies, 
 \begin{align}
\ushort R(d) &= 
h -  \log \sqrt{ 2 \pi e d} 
\label{eq:slbmse}
\\
 &
 \leq 
 R(d) 
 \label{eq:slblim} 
 \\
 &\leq L(d),
 \end{align}
 where 
\begin{equation}
h \triangleq \limsup_{n \to \infty}  \frac 1 n h(X^n). \label{eq:h}
\end{equation}
Furthermore, suppose that the density $f_{X^n}$ is $c_1 \|x^n\| + c_0 \sqrt n$-regular with some $c_0 \geq 0$, $c_1 \geq 0$, and that there exists a constant $\alpha > 0$ such that
\begin{equation}
 \E{  \|X^n\| } \leq \sqrt n \alpha. \label{eq:normu}
\end{equation}
Then, as $d \to 0$, the lattice rate-distortion function is upper bounded by,
 \begin{align}
L(d) &\leq  \ushort R(d) + \bigo{\sqrt d}. \label{eq:L0}
\end{align}

\label{thm:Ldlim}
\end{thm}

\begin{proof}
The inequality in \eqref{eq:slblim} is obtained by applying \eqref{eq:RH} and the Shannon lower bound to the expression under the limsup in \eqref{eq:Rdlim}, and taking $n$ to infinity. The inequality in \eqref{eq:L0} is obtained by applying the bound
\begin{equation}
 D(X \| X_\cC) \leq 2 r_{\cC} \left( c_1\E{\|X\|} + c_1 r_\cC + c_0 \sqrt n \right) \log e, 
\end{equation}
which is a particularization of \eqref{eq:divu}, to 
\eqref{eq:Lnu}, normalizing by $n$ and taking a limsup in $n$.  
\end{proof}

\thmref{thm:Ldlim} establishes that for a wide class of sources with sufficiently smooth densities, which includes non-stationary and non-ergodic sources,  the lattice rate-distortion function approaches Shannon's lower bound at a speed $\bigo{\sqrt d}$ as $d \to 0$. 

\thmref{thm:Ldlim} implies that for sources with regular densities, 
\begin{equation}
 \lim_{d \to 0}   \limsup_{n \to \infty} \left[   \frac 1 n \mathbb L_{X^n}(d) + \log \sqrt d \right] = h  - \log \sqrt{ 2 \pi e}.  \label{eq:L0a}
\end{equation}
A weaker result, namely, 
\begin{equation}
 \limsup_{n \to \infty}  \lim_{d \to 0}   \left[  \frac 1 n \mathbb L_{X^n}(d) + \log \sqrt d \right] = h  - \log \sqrt{ 2 \pi e},  \label{eq:Lfnlim}
\end{equation}
can be obtained with a weaker assumption on the source distribution: for \eqref{eq:Lfnlim} to hold, only $H (\lfloor X^n \rfloor) < \infty$ is required. Indeed, applying Csisz\'ar's result \eqref{eq:kl} to \eqref{eq:Lnu} and taking the limit in $d$, we obtain
\begin{equation}
 \lim_{d \to 0}   \left[ \mathbb L_{X}(d) + n \log \sqrt d \right] = h(X)  - n \log \sqrt{ 2 \pi e} + \bigo{\log n}. \label{eq:zamir}
\end{equation}
Dividing by $n$ and taking $n$ to infinity leads to \eqref{eq:Lfnlim}. The reason Csisz\'ar's result in \thmref{thm:entropyq} is insufficient to prove \eqref{eq:L0a} is that even though it establishes that $D(X \| X_\cC)$ converges to $0$ as $d \to 0$ for any fixed $n$, it leaves unaddressed the behavior of $D(X \| X_\cC)$ as $n$ grows.

Equality in \eqref{eq:zamir} implies that for $X \in \mathbb R^n$ with $H(\lfloor X \rfloor) < \infty$, 
\begin{equation}
\lim_{d \to 0} \frac{ \mathbb L_{X}(d) }{ \frac{n}{2} \log \frac 1 {d}  } = 1, 
\end{equation}
which can be viewed as a lattice counterpart of R\'enyi information dimension \cite{renyi1959dimension}. 
\begin{apxonly}
\begin{equation}
d(X) =  \lim_{\Delta \to 0} \frac{H( \left \lfloor \frac X \Delta \right \rfloor )}{ \log \frac 1 {\Delta}}
\end{equation}
Renyi actually defined this with fineness parameter $m$, where $\Delta = \frac 1 m$. For continuous r.v.'s, 
\begin{equation}
d(X) = 1.
\end{equation} 
\end{apxonly}

While the result in \eqref{eq:L0a} is new, the statements similar to \eqref{eq:Lfnlim} and \eqref{eq:zamir} are found in the existing literature. 
The counterpart of \eqref{eq:Lfnlim} for dithered lattice quantization is contained in Zamir's text \cite{zamir2014lattice}. 
\apxonl{For non-dithered quantization with tessellating quantizers, a counterpart of \eqref{eq:Lfnlim} is shown by Linder and Zeger \cite{linder1994tessellating}. The validity of their proof depends on the validity of a conjecture by Gersho.} The following result was shown by Linkov \cite{linkov1965evaluation} and revisited, under progressively more general assumptions,  by Linder and Zamir \cite{linder1994asymptotic} and by Koch \cite{koch2015shannonlb}, who showed that as long as $H(\lfloor X \rfloor ) < \infty$, it holds that
\begin{equation}
 \lim_{d \to 0}   \left[ \mathbb R_{X}(d) + n \log \sqrt d  \right] = h(X)  - n \log \sqrt{ 2 \pi e}, \label{eq:linkov}
\end{equation}
where $X \in \mathbb R^n$, and $ \mathbb R_{X}(d)$ is the minimal mutual information quantity defined in \eqref{eq:RR(d)}. 
Regarding the operational meaning of \eqref{eq:linkov} in the context of $n$-dimensional quantization, we note the following observations, which highlight the difference between dithered and non-dithered quantization. 
\begin{itemize}
 \item 
 Koch and Vazquez-Vilar \cite{koch2015gibbs} recently showed that if one replaces $\mathbb R_{X}(d)$ in \eqref{eq:linkov} by the minimum output entropy attainable by an $n$-dimensional quantizer operating at average distortion $d$, then the resulting limit as $d \to 0$ is strictly greater than the right side of \eqref{eq:linkov}. 

\item The reasoning in \cite{linkov1965evaluation,linder1994asymptotic,koch2015shannonlb} reveals that
\begin{equation}
 \lim_{d \to 0}   \left[ I(X; X + Z) + n \log \sqrt d \right] = h(X)  - n \log \sqrt{ 2 \pi e}, \label{eq:linderdither}
\end{equation}
where the choice of $Z$ satisfies $\E{\sd(X, X+Z)} \leq d$. Since, operationally, $I(X; X + Z)$ corresponds to the quantization rate (see e.g. \cite{zamir2014lattice}) of $X$ dithered by $Z$, there exists an $n$-dimensional dithered quantizer operating at average distortion $d$ and whose rate satisfies  \eqref{eq:linderdither}.  
\end{itemize}

\begin{apxonly}
For non-lattice quantization of memoryless sources, the asymptotics in \eqref{eq:zamir} can be obtained via \eqref{eq:linkov} and \cite[Theorem 9]{kostina2015varrate} which implies \apxonl{ using \cite[ (125) and (148)]{kostina2015varrate} } that
\begin{equation}
\mathbb L_{X^n}^\star(d) \leq n R_{\mathsf X}(d) + \frac 1 2 \log n + \bigo{1}. 
\end{equation}
 
\end{apxonly}

\begin{remark}
If, instead of requiring that $\Prob{\frac 1 n \| X - q_\cC(X)\|^2 \leq d} = 1$  as in \eqref{eq:Ld}, we ask only that $\E{\frac 1 n \| X - q_\cC(X)\|^2} \leq d$, then the analog of the result in \eqref{eq:zamir} can be alternatively obtained as follows. Denote the minimum
 of normalized second moments over all $n$-dimensional lattices by 
\begin{equation}
G_n^\star \triangleq \min_{\cC_n}\frac {\E{ \| U_{\cC_n}\|} } {n V_{\cC_n}^{\frac 2 n} },
\end{equation}
where $U_{\cC_n}$ is uniform on $\cV_{\cC_n}(0)$. In \cite[(25)]{zamir1996onlatticenoise}, it is shown that $G^{\star}_n$
converges to $\frac 1 {2\pi e}$ at a rate
\begin{equation}
\frac 1 n \log(2\pi e G^\star_n) = \bigo{ \frac {\log n} n}, \label{eq:zamirfeder}
\end{equation}
a result which Zamir and Feder attributed to Poltyrev. Let $\cC_{n, d}^\star$ be the lattice whose normalized second moment equals $G_n^\star$ rescaled so that its mean
square error (with respect to $X^n$) is $d$. Using $\cC_{n, d}^\star$ in \cite[Theorem 1]{linder1994tessellating},  one
concludes that 
\begin{align}
\lim_{d \to 0}    \left( H \left( q_{\cC_{n, d}^\star}(X)\right)  + n \log \sqrt d  \right)
                          =   h(X) + \frac 1 2 \log (G^{\star}_n) . \label{eq:linderzeger}
\end{align}
Substituting \eqref{eq:zamirfeder} into
\eqref{eq:linderzeger}, one obtains the same asymptotics as in \eqref{eq:zamir}. A gentle modification of the above argument (apparent from  \cite[(26)]{zamir1996onlatticenoise} and \cite[Lemma 1]{linder1994tessellating}) leads to \eqref{eq:zamir} for the maximal distortion criterion as well. 
\label{rem:linder}
\end{remark}
\begin{apxonly}
Average MSE distortion over the quantization cell (per dimension): 
\begin{equation}
d_\cC =  \frac 1 {n V_\cC}\int_{\mathcal V_\cC(0)} \|x\|^2 dx
\end{equation}
For uniform quantizer $d_\cC = \frac 1 {12} \Delta^2$. If $\mathcal V_0$ is an $n$-dimensional ball of radius $\sqrt n \Delta$,  
\begin{equation}
d_\cC =  \frac{n \Delta^2}{n + 2}.
\end{equation}
As $d \to 0$, the rate achieved by the lattice quantizer with average MSE distortion (this is \eqref{eq:linderzeger}):
\begin{equation}
H \left( q_{\cC_{n, d}^\star}(X^n)\right) - \ushort{ \mathbb R}_{X^n} (d) = \frac 1 2 \log (2 \pi e G^\star_n) + o(1)
\end{equation}

Asymptotically suboptimum scheme:
perform uniform quantization on $\mathbb R$, then losslessly compress the quantized version. Putting the reconstruction points at the middle of the quantization intervals results in average MSE distortion $d = \frac{\Delta^2}{12}$ (because $\E{(X - c_j)^2 | X \in \Delta_j} \approx \frac 1 \Delta \int_{\Delta_j} (x - c_j)^2 dx =  \frac{\Delta^2}{12}$), so the rate achieved is
\begin{equation}
R_{uni}(d) = h(X) +\frac 1 2 \log \frac 1 d - \frac 1 2 \log 12
\end{equation}

The difference with the best vector quantizer is $\frac 1 2 \log \frac{\pi e}{6} \approx 0.25$. This is exactly $D(U[0, 1] \| \mathcal N(1/2, 1/4))$  (nongausianness of $U$). This is suboptimal. 

\end{apxonly}

\apxonl{Can we claim that
\begin{equation}
\mathbb R_X(d) \geq h(X) ?
\end{equation}
}

\subsection{Second order analysis}
\label{sec:asymp2}

The minimum achievable coding rate at a given blocklength and a given excess distortion probability is defined as
\begin{equation}
 R(n, d, \epsilon) \triangleq \frac 1 n \min\{ \log M \colon \exists (M, d, \epsilon) \text{ code for } X \in \mathbb R^n \}.
\end{equation}

\thmref{thm:2order}, stated next, provides a refined approximation to $R(n, d, \epsilon)$ at a given blocklength and a given low distortion. 

\begin{thm}
 Let $\mathsf X \in \mathbb R$ have $c_1 |\mathsf x| + c_0$-regular density $f_\sX$ such that  $\E{| \log f_{\sX}(\sX)|^3} < \infty$ and $\E{\sX^4} < \infty$. For the compression of the source consisting of i.i.d. copies of $\sX$ under the mean-square error distortion, it holds that 
\begin{equation}
 R(n,d,\epsilon) =  \ushort{R}(d) + \sqrt{\frac {\ushort{\mathcal V}} n }\Qinv{\epsilon} + O_1\left(\sqrt d\right) + O_2\left( \frac{\log n} n \right).  \label{eq:Rapprox0}
\end{equation}
where $\ushort{R}(d) $ and $\ushort{\mathcal V}$ are given by the mean and variance of 
\begin{equation}
\ushort \jmath_{\sX}(\sX, d) = \log \frac 1 {f_{\sX}(\sX)} - \log \sqrt{2\pi e d},
\end{equation}
 respectively, 
 and 
\begin{align}
 0 &\leq O_1\left(\sqrt d\right) \leq \bigo{\sqrt d}, \\
 \bigo{\frac 1 n } &\leq O_2 \left( \frac{\log n}{n}\right)  \\
 &\leq \log_2(2 \sqrt{\pi e}) \frac{\log n}{n} + \bigo{\frac 1 n \log \log n}.
\end{align}
 Furthermore,  \eqref{eq:Rapprox0} is attained by lattice quantization. 
\label{thm:2order}
\end{thm}

\begin{proof}[Proof of the converse part]
We show that as long as $\sX \in \mathbb R$ has a density (regularity is not required for the converse), the minimum rate required to quantize $X$, which is a vector of $n$ i.i.d. copies of $\sX$, is at least
\begin{align}
n R(n, d, \epsilon) &\geq  n \ushort{R}(d) + \sqrt{n \ushort{\mathcal V}} \Qinv{\epsilon}  +  \bigo{1}. \label{eq:C2}
\end{align}
The proof consists of the analysis of the converse bound in \eqref{eq:C}. Letting $\mu$ be the Lebesgue measure on $\mathbb R^n$ and letting $\sd$ be the mean-square error, observe that regardless of the choice of $y \in \mathbb R^n$, $\mu \left[ \sd(X, y) \leq d\right]$ is equal to the volume of Euclidean ball of radius $\sqrt {n d}$, i.e.
\begin{align}
 \log \mu \left[ \sd(X, y) \leq d\right]  &= \log b_n + n \log \sqrt{n d} \\
 &= n \log \sqrt {2 \pi e d}  - \frac 1 2 \log n + \bigo{1} \label{eq:muvol},
\end{align}
where to get \eqref{eq:muvol} we invoked \eqref{eq:Ana}. 
Furthermore, the Neyman-Pearson function expands as \cite[Lemma 58]{polyanskiy2010channel}, \cite[(251)]{kostina2011fixed}
 \begin{align}
\log \beta_{1 - \epsilon} (P_{X}, \mu) &= n h( \sX) + \sqrt {n \ushort{\mathcal V}} \Qinv{\epsilon} 
\notag\\
&
- \frac 1 2 \log n + \bigo{1}. \label{eq:beta2order} 
\end{align}
According to  \eqref{eq:C}, for any $(M, d, \epsilon)$ code, $\log M$ is lower bounded by the difference between \eqref{eq:beta2order} and \eqref{eq:muvol}, which is exactly \eqref{eq:C2}.
\end{proof}

\begin{proof}[Proof of the achievability part]
The proof consists of the analysis of the bound on the excess distortion probability of lattice quantizers in \eqref{eq:agamma}. 
Assume that $X$ is a vector of $n$ i.i.d. copies of $\sX$. According to \propref{prop:regprod}, the density of $X$ is $c_1 \|x\| + c_0 \sqrt n$-regular.

First, consider the case of non-uniform distribution: $\Var{f_\sX(\sX)} > 0$. 
The second term in \eqref{eq:agamma} is equal to
\begin{align}
 &~
 \Prob{ 2 r_{\cC} v_\cC(X) > \gamma} 
 \notag\\
 =&~  \Prob{ 2 \sqrt{n d} (c_1 \|X\| + c_1 \sqrt{n d} + c_0 \sqrt n) > \gamma} \label{eq:a2term}.
\end{align}
Denote the constant
\begin{equation}
\alpha \triangleq  \E{\sX^2}.  
\end{equation}
By Chebyshev's inequality, 
\begin{equation}
\Prob{ \|X\|^2 > 2 n \alpha} \leq \frac 1 n  \label{eq:cheb}.
\end{equation}
So, the choice
\begin{equation}
 \gamma = 2 n \sqrt d\left( c_1 \sqrt \alpha + c_1 \sqrt d + c_0 \right) \label{eq:gammaa}
\end{equation}
ensures that the probability in \eqref{eq:a2term} is upper bounded by $\frac 1 n$.  To analyze the the first term in \eqref{eq:agamma}, note that according to the Berry-Esse\'en theorem, for all $0 < \epsilon^\prime < 1$, 
\begin{align}
 &~ \Prob{\log \frac 1 {f_{X}(X)} > n h(\sX) +  \sqrt{n \Var{\log f_{\sX}(\sX)}  } \Qinv{\epsilon^\prime}} \notag\\
 \leq&~ \epsilon^\prime + \frac {B}{\sqrt n},
\end{align}
where 
\begin{align}
B = 6 \frac{\E{| \log f_\sX(\sX) + h(\sX)|^3}}{\Var{\log f_\sX(\sX)}} \label{eq:be}
\end{align}
 is the Berry-Esse\'en constant, finite by the assumptions $\Var{\log f_\sX(\sX)} > 0$ and $\E{| \log f_{\sX}(\sX)|^3} < \infty$. Therefore, letting
\begin{align}
&~
\log M  \label{eq:Ma} \\
=&~ n h(\sX) - \log V_\cC + \sqrt{n \Var{\log f_{\sX}(\sX)}  } \Qinv{\epsilon^\prime}  + \gamma, \notag
\end{align}
 we conclude that
\begin{equation}
\Prob{\log \frac 1 {f_{X}(X)} - \log V_\cC + \gamma  > \log M } \leq  \epsilon^\prime + \frac {B}{\sqrt n}.
\end{equation}
Finally, choosing $\epsilon^\prime$ as
\begin{equation}
\epsilon^\prime = \epsilon -  \frac {B}{\sqrt n} - \frac 1 n, \label{eq:epsa}
\end{equation}
we conclude that the sum of both terms in \eqref{eq:agamma} does not exceed $\epsilon$. It follows that there exists an $(M, d, \epsilon)$ code with $M$ given in \eqref{eq:Ma} and $\epsilon$ given in \eqref{eq:epsa}. Letting $\cC$ be a lattice satisfying \eqref{eq:rogers} and applying \eqref{eq:Vd}, \eqref{eq:Ana} and \eqref{eq:rogers} to \eqref{eq:Ma}, we express \eqref{eq:Ma} as
\begin{align}
&~ \log M  \notag\\
=&~ n h(\sX) -  n \log \sqrt {2 \pi e d} + \sqrt{n \Var{\log f_{\sX}(\sX)}  } \Qinv{\epsilon} \notag \\
+&~ \log_2(2 \sqrt{\pi e}) \log n + \bigo{\log \log n} + \bigo{\sqrt d},
\end{align}
 which concludes the proof of the achievability part of \eqref{eq:Rapprox0} for non-uniform $\sX$.

 If $\sX$ is uniform on a compact set, then $\log \frac 1 {f_{X}(X)} = n h(\sX)$ a.s.,  $v_\cC(x) \equiv 0$, and \eqref{eq:agamma} implies that there exists an $(M, d, \epsilon)$ code with 
\begin{equation}
 \epsilon = \1{n h(\sX) - \log V_\cC > \log M }.
\end{equation}
 Choosing
\begin{equation}
\log M  = n h(\sX) - \log V_\cC
\end{equation}
results in $\epsilon = 0$. It follows that if $\sX$ is uniform, then there exists an $(M, d, 0)$ code with
\begin{align}
\log M  &= n h(\sX) -  n \log \sqrt {2 \pi e d}   
\notag \\
&
+ \log_2(2 \sqrt{\pi e}) \log n + \bigo{\log \log n} .
\end{align}

\begin{apxonly}
 Chebyshev: 
\begin{equation}
\Prob{|\|X\|^2 - \E{\|X\|^2}| > \sqrt {n \Var{\|X\|^2}} } \leq \frac 1 {n}
\end{equation}

\begin{align}
\E{\|X\|^2} &= n \E{\sX^2} \\
\Var{\|X\|^2} &= n \Var{\sX^2} 
\end{align}

so 
\begin{equation}
\Prob{ \|X\|^2 > n \left( \E{\sX^2} + \sqrt{\Var{\sX^2}} \right)} \leq \frac 1 n  
\end{equation}
\end{apxonly}

\end{proof}

It can be shown \cite{kostina2016slb} that \eqref{eq:C2} continues to hold  for finite alphabet sources. The $\bigo{1}$ lower bound on the third order term in  \eqref{eq:C2} presents an improvement for the cases where Shannon's lower bound is tight over the general $ - \frac 1 2 \log n + \bigo{1}$ lower bound shown in \cite{kostina2011fixed} .

We conclude this section with a result that provides an estimate of the speed of convergence to $\ushort R(d)$ for sources with memory. At this level of generality, even a first order asymptotic analysis is highly nontrivial, and no second-order results exist to date.  \thmref{thm:2ordermem} below shows that for a class of  sources with memory, the rate of approach to the rate-distortion function is of order $\frac 1 {\sqrt n}$. Even though \thmref{thm:2ordermem} does not specify the constant in front of $\frac 1 {\sqrt n}$, it presents a step forward in the notoriously difficult problem of quantifying the rate-distortion tradeoffs for sources with memory. \thmref{thm:2ordermem} is an easy implication of the approach developed in \secref{sec:slb} and \secref{sec:lattice}. 

\begin{thm}
 Let the random process $X_1, X_2, \ldots$ be such that the density $f_{X^n}$ is log-concave and $c_1 \|x^n\| + c_0 \sqrt n$-regular with some $c_0 \geq 0$, $c_1 \geq 0$, and that the expectation of the norm of $X^n$ is bounded as in \eqref{eq:normu}. For the compression of $X_1, X_2, \ldots$ under mean-square error distortion, it holds that 
\begin{equation}
 R(n,d,\epsilon) =  \ushort{R}(d) + \frac {q(\epsilon)} {\sqrt n}  + \bigo{\sqrt d} + \bigo{ \frac{\log n} n },  \label{eq:Rapprox0mem}
\end{equation}
where $\ushort{R}(d) $ is given in \eqref{eq:slbmse}, and 
\begin{equation}
-\sqrt{\frac 1 {1 - \epsilon}} \leq q(\epsilon) \leq \sqrt{\frac 1 {\epsilon}}. 
\end{equation}
Moreover, \eqref{eq:Rapprox0mem} is attained by lattice quantization. 
\label{thm:2ordermem}
\end{thm}

\begin{proof}[Proof of the converse part]
We weaken \eqref{eq:Ca} by choosing $\gamma = \frac 1 2 \log n$ and letting $\mu$ be the Lebesgue measure to deduce that the parameters of any $(M, d, \epsilon)$ code must satisfy the inequality
\begin{align}
\epsilon 
\geq &~ 1 - \Prob{ \log \frac 1 {f_X(X)} - \phi(d) < \log M + \frac 1 2 \log n} - \frac 1 {\sqrt n}. \label{eq:cmem}
\end{align}
For
\begin{equation}
\log M =  h(X) - \phi(d) -  \sqrt{ \frac{\Var{f_X(X)}}{1 - \epsilon^\prime - \frac 1 {\sqrt n}}} - \frac 1 2 \log n,
\end{equation}
we observe that due to Chebyshev's inequality, the probability in the right side of \eqref{eq:cmem} is upper-bounded by $1 - \epsilon^\prime - \frac 1 {\sqrt n}$. We conclude that $\epsilon \geq \epsilon^\prime$, which implies the validity of the converse part of \eqref{eq:Rapprox0mem} when combined with the recent result of Fradelizi et al. \cite[Theorem 2.3]{fradelizi2015optimal},  which states that as long as $X \in \mathbb R^n$ has log-concave density, 
\begin{equation}
 \Var{\log f_X(X)} \leq n.  \label{eq:fradelizi}
\end{equation}

\end{proof}

\begin{proof}[Proof of the achievability part]
The proof mimics the proof of the achievability part of \thmref{thm:2order}, replacing the application of the Berry-Esse\'en theorem by  Chebyshev's inequality. Namely, due to \eqref{eq:normu}, \eqref{eq:cheb} continues to hold. By Chebyshev's inequality, 
\begin{equation}
 \Prob{ \log \frac 1 {f_X(X)} - h(X) \geq  \sqrt{ \frac{\Var{f_X(X)}}{\epsilon^\prime}} } \leq \epsilon^\prime. \label{eq:chebf}
\end{equation}
Letting
\begin{align}
\log M &= h(X) - \log V_\cC + \sqrt{ \frac{\Var{f_X(X)}}{\epsilon^\prime}}  + \gamma, \\
\epsilon^\prime &= \epsilon - \frac 1 n,
\end{align}
where $\gamma$ is chosen as in \eqref{eq:gammaa}, we conclude that the sum of both terms in \eqref{eq:agamma} does not exceed $\epsilon$. 
The proof is complete upon applying \eqref{eq:fradelizi}. 
\end{proof}

\section{Beyond MSE distortion}
\label{sec:beyondMSE}
\secref{sec:lattice} discussed lattices that are good for covering with respect the Euclidean norm, and accordingly, the asymptotic analysis in \secref{sec:asymp} focused on the mean-square error distortion. This section summarizes how to generalize those results to a wider class of distortion measures.  

We consider distortion measures of form 
\begin{equation}
\mathsf d(x, y) = \mathsf d( n^{- \frac 1 p} \| \mathsf W (x - y) \|_p) \label{eq:dnorm}
\end{equation}
where  $\mathsf W$ is an $n \times n$ invertible matrix, $\|\cdot \|_p$ is the $L^p$ norm in $\mathbb R^n$, $1 \leq p \leq \infty$, and $\mathsf d \colon \mathbb R^+ \mapsto \mathbb R^+$ is right-continuous. The scaling by $n^{-\frac 1 p}$ in \eqref{eq:dnorm} is chosen so that the distortion does not have a tendency to increase with increasing dimension $n$.

\begin{example}
Scaled weighted $L^p$ norm distortion fits the framework of \eqref{eq:dnorm}:  
 \begin{equation}
\mathsf d(x, y) = n^{- \frac s p} \| \mathsf W( x - y) \|^s_p \label{eq:dLps},
\end{equation}
where $s > 0$. 
Plugging $s = 2$ and $p = 2$ in \eqref{eq:dLps} one recovers the MSE distortion measure. 
An interesting special case is that of the $L^\infty$ norm, which corresponds to the distortion measure 
\begin{equation}
 \mathsf d(x^n, y^n) = \max_{1 \leq i \leq n} |x_i - y_i|^s. \label{eq:Linf}
\end{equation}
\end{example}

\todo{Rewrite the weighted MSE stuff by changing the source. rather than the distortion measure.} 
 \begin{example}
Weighted MSE distortion measure also fits the framework of \eqref{eq:dnorm}: 
 \begin{equation}
\mathsf d(x, y) =  \frac 1 n \|  \mathsf W(x - y) \|_2^2 \label{eq:dnormMSEW}.
\end{equation}
\end{example}



Defining the nearest-neighbor quantizer and the lattice covering radius in terms of weighted (by $\mathsf W$) $L^p$ norm, rather than the Euclidean norm, one can generalize Section \ref{sec:lattice} to distortion measures of type \eqref{eq:dnorm}. 
The maximum distortion is related to the (weighted $L^p$) covering radius as
\begin{align}
r_\cC &= n^{\frac 1 p} r(d),\\
r(d) &\triangleq  \inf\{ r \geq 0 \colon \sd(r) \leq d\} \label{eq:rn}.
\end{align}
If $\sd \colon \mathbb R^+ \mapsto \mathbb R^+$ is invertible, then  simply $r(d) = \sd^{-1}\left(  d\right)$. For example,  the distortion measure in \eqref{eq:dLps} corresponds to $\mathsf d\left( r\right) = r^s$; therefore, $r(d) = \sqrt[s]{d}$. The lattice cell volume can be expressed as
\begin{align}
 \log V_\cC 
 &= n \log  r_\cC + \log b_{n,p} + \log |\det \mathsf W| - n \log \rho_{\cC}, \label{eq:Vdp}
\end{align}
where $b_{n, p}$ is the volume of a unit $L^p$ ball:
\begin{equation}
b_{n, p} \triangleq \frac{ \left( 2 \Gamma \left(\frac 1 p +1\right) \right)^n }{ \Gamma \left(\frac n p+1 \right) }. \label{eq:bnp}
\end{equation}

A curious special case is that of $L^\infty$ norm, which corresponds to the distortion measure in \eqref{eq:Linf}: since an $L^\infty$ ball is simply a cube, the  cubic lattice quantizer attains the best covering efficiency $\rho_{\cC} = 1$.

Substituting \eqref{eq:Vdp} into \eqref{eq:renyig1}, we express the entropy at the output of the weighted $L_p$ quantizer based on lattice $\cC$ as
\begin{align}
 H\left( q_\cC(X) \right) &= h(X) - n \log  r_\cC - \log b_{n,p} - \log |\det \mathsf W| 
 \notag\\
 &+ n \log \rho_{\cC} +  D(X \| X_\cC). \label{eq:hqup}
\end{align}

By Stirling's approximation, as $n \to \infty$, \eqref{eq:bnp} expands as
\begin{align}
 \log  b_{n, p} &=  n \log c_p - \frac n p \log n  - \frac 1 {2} \log n+ \bigo{1 } \label{eq:anp},\, p < \infty,\\
 \log  b_{n, \infty} &=  n \log c_\infty = 2 n, \label{eq:anpinf}
\end{align}
where 
\begin{align}
c_p &\triangleq 2 \Gamma \left( \frac 1 p + 1 \right) (p e)^\frac 1 p, \quad p < \infty,\\
c_\infty &\triangleq \lim_{p \to \infty} c_p = 2.
\end{align}

To study the covering efficiency of lattices with respect to $\sd$, we invoke the following result of Rogers to complement Rogers' Theorem \ref{thm:rogers}:   
\begin{thm}[{Rogers \cite[Theorem 5.8]{rogers1964packing}, generalization of Theorem \ref{thm:rogers}}]
For each $n \geq 3$, there exists an $n$-dimensional lattice $\cC_n$  with covering efficiency (with respect to any norm)
\begin{equation}
n \log \rho_{\cC_n} \leq \log n \left( \log_2 n + c \log \log n \right)  ,
\end{equation}
where $c$ is a constant.
\label{thm:rogers2}
\end{thm}
When particularized to the Euclidean norm, Theorem \ref{thm:rogers2} presents a weakened version of Theorem~\ref{thm:rogers}. 

It follows from \eqref{eq:hqup}, \eqref{eq:anp}, \eqref{eq:anpinf} and \thmref{thm:rogers2} that for $1 \leq p \leq \infty$, the lattice $d$-entropy with respect to the distortion measure $\sd$ satisfies, as $n \to \infty$, 
\begin{align}
 \mathbb L_{X}(d) &\geq  h(X) - n  \log r(d) - n \log c_p - \log |\det \mathsf W|
 \notag\\
  &+ \frac 1 2 \log n + \bigo{1}  \label{eq:Lnpl},\\
  \mathbb L_{X}(d) &\leq  h(X) - n  \log r(d) - n \log c_p - \log |\det \mathsf W| \notag\\
  &+ D(X \| X_\cC) + \bigo{\log n  }  \label{eq:Lnpu}.
\end{align}

Plugging $r(d) = \sqrt{d}$, $c_2 = \sqrt {2 \pi e}$ and $\mathsf W = \mathsf I$ into \eqref{eq:Lnpl} and \eqref{eq:Lnpu}, one recovers the corresponding bounds for the mean-square error distortion, namely, \eqref{eq:Lnl} and \eqref{eq:Lnu}. 

It is enlightening to compare \eqref{eq:Lnpl} with Shannon's lower bound. 
For the distortion measure in \eqref{eq:dLps}, a direct calculation using Table \ref{tab:slbdiff} shows that Shannon's lower bound  is given by, for  $n \to \infty$, 
\begin{align}
 \ushort {\mathbb R}_{X} (d) &= h(X) + \frac n s \log \frac 1 d - \frac n p \log n - \log b_{n, p}  + \frac n s \log \frac{n }{ s e} 
 \notag\\
 &- \log \Gamma\left( \frac n s + 1\right) - \log |\det \mathsf W| \\
 &= h(X)  + \frac n s \log \frac 1 d  - n \log c_p - \log |\det \mathsf W| + \bigo{1}, \label{eq:uRs}
\end{align}
which up to the terms of order $\frac 1 2 \log n + \bigo{1}$ is the same as a particularization of \eqref{eq:Lnpl} to the distortion in \eqref{eq:dLps}. 
More generally, 
if $\sd(\cdot)$ is differentiable at $0$ and $0 < \sd^\prime(0)< \infty$, then by Taylor's approximation,
\begin{equation}
r(d) = \frac{d}{\sd^\prime(0)} + \smallo{d} \label{eq:r0}.
\end{equation}
If $\sd^\prime(0) = \ldots =  \sd^{(s-1)}(0) = 0$, and $0 < \sd^{(s)}(0) < \infty$, then
\begin{equation}
 r(d) = \sqrt[s]{\frac{ s!\, d}{\sd^{(s)}(0)}} + \smallo{\sqrt[s] d} \label{eq:r1}.
\end{equation}
Suppose further that $\mathsf d(\cdot)$ in the right side of \eqref{eq:dnorm} satisfies Linkov's regularity conditions \eqref{lin:a}--\eqref{lin:c}. Then, \cite[Corollaries 1, 2]{linkov1965evaluation} imply that
\begin{align}
\ushort {\mathbb R}_{X} (d) &= h(X) +  \frac n s \log \frac {\sd^{(s)}(0)} {s! d} - n \log c_p - \log |\det \mathsf W| \notag\\
&+ n\, \smallo{1} + \bigo{1}, \label{eq:linkova} 
\end{align}
where $\smallo{1}$ denotes a term that vanishes (uniformly in $n$) as $d \to 0$, and $\bigo{1}$ denotes a term that is bounded by a constant. Again, up to the remainder terms, this coincides with \eqref{eq:Lnpl}.

Next, we study the behavior of \eqref{eq:Lnpu}, which requires the following notion of regularity with respect to a weighted $L_q$ distance: a differentiable probability density function $f_X$ is $v$-regular if
\begin{equation}
\| \mathsf W^{-1} \nabla f_X(x)\|_q \leq v(x) f_X(x),  \qquad \forall x \in \mathbb R^n. 
\label{eq:pwregq}
\end{equation}
\thmref{thm:hq} generalizes as follows:
\begin{thm}
Let $1 \leq p \leq \infty$ and let $\frac 1 p + \frac 1 q = 1$. Let $X$ be a random variable with $h(X) > -\infty$ and $v$-regular density (according to \eqref{eq:pwregq}). Let $\mathcal C$ be a lattice in $\mathbb R^n$. Then the information random variable and the entropy at the output of lattice quantizer formed for the distortion in \eqref{eq:dnorm} can be bounded as \eqref{eq:iqv} and \eqref{eq:hq}, respectively, with $V_\cC$ in \eqref{eq:Vdp}. 
Furthermore, if $v(x) = v(\|\mathsf W x\|_p)$ is convex and nondecreasing, then \eqref{eq:iqv} and \eqref{eq:hq} can be strengthened by replacing \eqref{eq:vv} with
\begin{equation}
v_{\cC}(x) =   \frac 1 2 v(\|\mathsf W x\|_p) + \frac 1 2 v(\|\mathsf W x\|_p + 2 r_\cC) \label{eq:v2p},
\end{equation}
and  if $v(x) = c_1 \|\mathsf W x\|_p + c_0$, then \eqref{eq:v2p} particularizes as 
\begin{equation}
v_{\cC}(x) = c_1 \|\mathsf W x\|_p + c_1 r_\cC + c_0 \label{eq:vpwp}.
\end{equation}
\label{thm:hqp}
\end{thm}

\begin{proof}
 The reasoning leading up to \eqref{eq:pra} is adjusted as:  
\begin{align}
&~ \left|  \log f_X(a) - \log f_X(b)\right| \notag \\
 =&~   \left| \int_0^1 \left( \mathsf W^{-1} \nabla \log f_X(t a + (1 - t)b), \mathsf W (a - b) \right) dt \right|\\
 \leq&~  \|\mathsf W (a - b) \|_p \log e  \int_0^1 v( t a + (1 - t)b ) dt  \label{eq:praap}\\
 \leq&~ \max_{0 \leq t \leq 1}  v(t a + (1 - t)b) \| \mathsf W (a - b) \|_p  \log e \label{eq:prap},
\end{align}
where \eqref{eq:praap} is by H\"older's inequality. The proof of \eqref{eq:v2p} and \eqref{eq:vpwp} is identical to the proof of \eqref{eq:v2} and \eqref{eq:vpw}.
\end{proof}

We are now prepared to state the generalizations of the asymptotic results in \secref{sec:asymp} to non-MSE distortion measures. 

Theorem \ref{thm:Ldlim} generalizes to the distortion measure in \eqref{eq:dnorm} as follows.

\begin{thm}
 Consider a random process $X_1, X_2, \ldots$ and a sequence of distortion measures given by  \eqref{eq:dnorm} with $\mathsf W = \mathsf W_n$ such that the limit
\begin{equation}
\omega \triangleq \lim_{n \to \infty} \frac 1 n \log |\det \mathsf W_n|
\end{equation}
exists and finite.

The lattice rate-distortion function satisfies, 
 \begin{align}
h - \log r(d)  - \log c_p - \omega 
 &\leq L(d), \label{eq:Llp}
 \end{align}
 and $h$ is defined in \eqref{eq:h}. 
Furthermore, suppose that $p \geq 2$ and that the density $f_{X^n}$ is $c_1 \| \mathsf W_n x^n\|_p + c_0 n^{\frac 1 p}$-regular with some $c_0 \geq 0$, $c_1 \geq 0$, and that there exists a constant $\alpha > 0$ such that
\begin{equation}
 \E{  \|\mathsf W_n X^n\|_p } \leq n^{\frac 1 p} \alpha. \label{eq:normup}
\end{equation}
Then, as $d \to 0$, the lattice rate-distortion function is upper bounded by,
 \begin{align}
L(d) &\leq  h -  \log r(d) - \log c_p - \omega + \bigo{r(d)} .\label{eq:L0p}
\end{align}

\label{thm:Ldlimp}
\end{thm}

\begin{proof}
The lower bound in \eqref{eq:Llp} follows from \eqref{eq:Lnpl}. To show \eqref{eq:L0p}, we apply the bound
\begin{align}
&~ D(X \| X_\cC) \\
\leq&~ 2 n^{\frac 1 p} r(d) \left( c_1\E{\|\mathsf W_n X\|_p} + c_1 n^{\frac 1 p} r(d) + c_0 n^{\frac 1 p}\right) \log e, \notag
\end{align}
which is a particularization of \eqref{eq:divu}, to 
\eqref{eq:Lnu}, we normalize by $n$ and we take a limsup in $n$.  
\end{proof}

Using \eqref{eq:uRs}, we see that for the distortion measure in \eqref{eq:dLps}, Shannon's lower bound is given by 
\begin{equation}
\ushort R(d) = h + \frac 1 s \log \frac 1 d  - \log c_p - \omega \label{eq:slbs},
\end{equation}
which coincides with \eqref{eq:Llp}. More generally, if $\sd(\cdot)$ satisfies Linkov's conditions \eqref{lin:a}--\eqref{lin:c}, observe using  \eqref{eq:r0}, \eqref{eq:r1} and \eqref{eq:linkova} that 
\begin{equation}
 \ushort R(d) =  h - \log r(d)  - \log c_p - \omega + \smallo{1} ,\quad d \to 0. \label{eq:L0cor2}
\end{equation}
It follows from \thmref{thm:Ldlimp} that for a large class of distortion measures, the lattice rate-distortion function approaches the Shannon lower bound as $d \to 0$:
\begin{equation}
 L(d) =  \ushort R(d) + \smallo{1} ,\quad d \to 0. \label{eq:L0cor2}
\end{equation}

\thmref{thm:2order} generalizes as follows. 
\begin{thm}[Generalization of Theorem \ref{thm:2order}]
 Let $p \geq 2$. Assume that the density of $\sX$ satisfies, 
\begin{equation}
|f^\prime_\sX(\mathsf x)| \leq  (c_1 |\mathsf x| + c_0) f_\sX(\mathsf x), \quad \forall x \in \mathbb R, \label{eq:reg1}
\end{equation}
where $c_0 \geq 0$ and $c_1 \geq 0$, and that 
   $\E{| \log f_{\sX}(\sX)|^3} < \infty$ and $\E{|\sX|^{2q}} < \infty$, where $\frac 1 p + \frac 1 q = 1$. Consider a sequence of distortion measures of type \eqref{eq:dnorm} with $\sd(\cdot)$ satisfying Linkov's conditions \eqref{lin:a}--\eqref{lin:c}, and $\mathsf W = \mathsf W_n$ is such that
\begin{equation}
\frac 1 n \log \left | \det \mathsf W_n \right| = \omega  + \bigo{ \frac{\log n} n },
\end{equation}
for some $\omega \in \mathbb R$, and that the minimum singular value of $\mathsf W_n$ is bounded below by some $\sigma > 0$.  
 For the compression of the source consisting of i.i.d. copies of $\sX$ under such distortion measure, it holds that 
\begin{equation}
 R(n,d,\epsilon) =  \ushort{R}(d) + \sqrt{\frac {\ushort{\mathcal V}} n }\Qinv{\epsilon} + \smallo{1} + \bigo{ \frac{\log n} n },  \label{eq:Rapprox1}
\end{equation}
where 
and $\ushort{R}(d)$ and $\ushort{\mathcal V}$  are given by the mean and the variance of
\begin{equation}
\ushort \jmath_{\sX}(\sX, d) = \log \frac 1 {f_\sX(\sX)} - \log r(d) - \log c_p - \omega, 
\end{equation}
respectively, 
and $\smallo{1}$ denotes a term that vanishes uniformly in $n$ as $d \to 0$. For $\mathsf d$ in \eqref{eq:dLps}, $\smallo{1}$ can be refined to $\bigo{\sqrt[s]{d}}$. Lattice quantization attains \eqref{eq:Rapprox1}. 
\label{thm:2orderg}
\end{thm}
\begin{proof}
The only observation required for the proof of \thmref{thm:2order} to apply is the following. If $X$ is a vector of $n$ i.i.d. copies of $\sX$, by \eqref{eq:reg1} and \propref{prop:regprod} it holds that \footnote{Note that \propref{prop:regprod} applies to any norm.} 
\begin{equation}
  \left\| \nabla f_X(x)\right\|_q \leq \left( c_1 \|x\|_q + c_0 n^{\frac 1 q} \right)  f_X(x).
\end{equation}
It follows that
\begin{equation}
 \left\| \mathsf W_n \nabla f_X(x)\right\|_q \leq (c_1 \|\mathsf W_n x\|_q + c_0 n^{\frac 1 q}) \sigma^{-2} f_X(x),
\end{equation}
that is, $X$ has a regular density in the sense of \eqref{eq:pwregq}, and \thmref{thm:hqp} can be applied in the same manner  \thmref{thm:hq} is used in the proof of \thmref{thm:2order}.  
\end{proof}
Note that \eqref{eq:Rapprox1} does not require the distortion measure to be separable.

\section{Conclusion}
Shannon's lower bound provides a powerful tool to study the rate-distortion function. We started the discussion by presenting an abstract Shannon's lower bound in \thmref{thm:slb} and its nonasymptotic analog in \thmref{thm:C}.  Theorem \ref{thm:lbeq} states the necessary and sufficient conditions for the Shannon lower bound to be attained exactly.  According to Pinkston's \thmref{thm:pinkston}, all finite alphabet sources satisfy that condition for a range of low distortions. Whenever the Shannon lower bound is attained exactly, the $\mathsf d$-tilted information in $x$ also admits a simple representation as the difference between the information in $x$ and a term that depends only on tolerated distortion $d$ (see \eqref{eq:jlbeq}).  This implies in particular that the rate-dispersion function of a discrete memoryless source with a balanced distortion measure is given simply by the varentropy of the source, as long as the target distortion is low enough.

\apxonl{it is known that the difference between Shannon's lower bound and the (informational) rate-distortion function vanishes with decreasing distortion \cite{linkov1965evaluation}; moreover, dithered lattice quantization can approach Shannon's lower bound \cite{linder1994tessellating,zamir2014lattice}. 
  This motivated a closer investigation of the performance non-dithered lattice quantization undertaken in this paper.  }

Although continuous sources rarely attain Shannon's lower bound exactly, they often approach it closely at low distortions.  
For a class of sources whose densities satisfy a smoothness condition, \thmref{thm:hq} presents a new bound on the output entropy of lattice quantizers in terms of the differential entropy of the source and the size of the lattice cells. The gap between the lattice achievability bound in \thmref{thm:hq} and the Shannon lower bound can be explicitly bounded in terms of the target distortion, the source dimension and the lattice covering efficiency.  \thmref{thm:hq} also presents a bound on the information random variable at the output of lattice quantizers. That latter bound is particularly useful for quantifying the nonasymptotic fundamental limits of lattice quantization. 

Leveraging the bound in \thmref{thm:hq}, we evaluated the best performance theoretically attainable by variable-length lattice quantization of general (i.e. not necessarily ergodic or stationary) real-valued sources in the limit of large dimension (Theorem \ref{thm:Ldlim}).  For high definition quantization of stationary memoryless sources whose densities satisfy a smoothness condition, we showed a Gaussian approximation expansion of the minimum achievable source coding rate (Theorem \ref{thm:2order}). The appeal of the new expansion is its explicit nature and a simpler form compared to the more general result in \cite{kostina2011fixed}. Going beyond memoryless sources, we showed that for a class of sources with memory, the Shannon lower bound is attained at a speed $\bigo{\frac 1 {\sqrt n}}$ with increasing blocklength (Theorem \ref{thm:2ordermem}). The engineering implication is that as long as the dimension $n$ is not too small and the target distortion is not too large, a separated architecture of a lattice quantizer followed by a lossless coder displayed in Fig. \ref{fig:latticelossless} is nearly optimal. Using a lattice with covering efficiency $\rho_{\cC}$ induces a penalty of 
$\log \rho_{\cC}$ to the attainable nonasymptotic coding rate. If the simplest uniform scalar quantizer is used, the penalty due to its covering inefficiency is still only $ \frac 1 2 \log \frac{\pi e}{2} \approx 1.05$ bits per sample. 


\section{Acknowledgement}
I would like to thank the Simons Institute for the Theory of Computing in Berkeley for providing the nurturing environment where this work started amid the discussions in the Spring of 2014; Dr. Tam\'as Linder for his insightful comments and in particular for describing an alternative way to obtain \eqref{eq:zamir}, now included in Remark \ref{rem:linder}; Dr. Tobias Koch for pointing out the references \cite{koch2015shannonlb,koch2015gibbs}; Dr. Yury Polyanskiy for mentioning the result in \cite{fradelizi2015optimal}.

{

\appendices

\detail{
\section{An upper bound on the entropy of an integer vector} 
\label{appx:entropymax}

We show an explicit bound on $H(\lfloor X \rfloor)$ in \eqref{eq:Hnub} that allows one to recover the sufficient conditions of  Wu and Verd\'u \cite{wu2010renyi} and Koch \cite{koch2015shannonlb}: 
\begin{equation}
 H(\lfloor X \rfloor) < \E{\log R_n\left(\| \lfloor X\rfloor  \|^2\right)}   + \log e +  \log\left( \E{\log R_n\left(\| \lfloor X \rfloor \|^2\right)} + 1\right) , \label{eq:Hnub}
\end{equation}
where $R_n(\sigma^2)$ is the number of integer-valued $n$-vectors whose Euclidean norm does not exceed $ \sigma$. In particular, for $n = 1$,  $R_1(k) =1 + 2k$, so 
\begin{equation}
 \E{\log R_1\left(\| X  \|^2\right)} = \E{ \log (  2 | \lfloor X\rfloor | + 1 )} 
\end{equation}
More generally, for $n \geq 1$ (Appendix \ref{appx:entropymax}),
\begin{equation}
\E{ \log R_n(\|\lfloor X \rfloor \|^2)} \leq  n \E{ \log \left( 1 + \frac 1 {\sqrt n}  \|X\|   \right)}  + n \log  \sqrt{2 \pi e} - \frac 1 2 \log (n + 2) - \log \sqrt \pi + \frac {2}{n + 2} \log e. \label{eq:Elogrn}
\end{equation}
 Bounds \eqref{eq:Hnub} and \eqref{eq:Elogrn} imply that 
a sufficient condition for \eqref{eq:entropy} is $\E{\log(1 + \frac 1 {\sqrt n} \|X\|)} < \infty$.

Assume that vector $X$ has all integer components. 
The following bound is shown in \cite[Lemma 3]{alon1994lower}:
\begin{align}\label{eq:aaa}
		H(X)  &\le L^\star_X  + \log_2 ( L^\star_X+1) + \log_2 e, 
\end{align}
where $L^\star_X$ is the minimum average length of a lossless representation of $X$.  Consider the following suboptimal lossless code:
enumerate all integer $n$-vectors so that 
\begin{equation}
\vect{j} \preceq \vect{k} \text{ if and only if } \| \vect{j}\| \leq \| \vect{k} \|,
\end{equation}
and let the lossless representation of $X$ be the binary representation of the index number of $X$.  Bounding $L^\star_X$ from above by the average length of this code, we have  
\begin{align}
 L^\star_X \leq \E{ \log_2 R_n \left( \|X\| \right) },
\end{align}
where $R_n(r)$ is the number of integer-valued vectors with norm not exceeding $r$:
\begin{equation}
 R_n(r) \triangleq \#\left\{ \vect k \in \mathbb Z^n \colon \| \vect k\| \leq r\right\}.
\end{equation}
To estimate $R_n(r)$, observe that $R_n(r)$ is upper bounded by the number of cubes with side $1$ that intersect the ball of radius $r$. Since the diameter of cube of side $1$ is $\sqrt n$, all those cubes are contained inside the augmented ball of radius $r + \sqrt n$. The number of volume-1 cubes that can be packed inside a ball is upper bounded by the volume of that ball; therefore,
\begin{equation}
 R_n(r) \leq b_n \left( r + \sqrt n \right)^n. \label{eq:Rnsu}
\end{equation}

A nonasymptotic Stirling's formula implies
\begin{equation}
\log b_n \leq \frac n 2 \log \frac {2 \pi e} {n + 2} - \frac 1 2 \log (n + 2) + \log \frac {e}{\sqrt \pi}, \label{eq:Anu}
\end{equation}
 so applying $\frac{z}{1 + z} \leq \log_e (1 + z)$ for $z > -1$ and \eqref{eq:Anu} to \eqref{eq:Rnsu}, we obtain
\begin{align}
\log R_n\left(r\right) &\leq  \log b_n - \frac n 2 \log n + n \log \left( \frac{\sqrt \sigma}{\sqrt n} + 1\right) \\
&\leq n \log \left( \frac{r}{\sqrt n} + 1\right)  + n \log  \sqrt{2 \pi e} - \frac 1 2 \log (n + 2) - \log \sqrt \pi + \frac {2}{n + 2} \log e \label{eq:rnu}. 
\end{align}
}

\begin{apxonly}
A WEAKER BOUND: 

For an integer-valued $n$-vector $\vect{k}$, define the following $n$-dimensional generalization of the geometric distribution function: 
\begin{equation}
G_p (\vect{k}) = \frac 1 {r_n(\| \vect{k}\|^2 )}p (1 - p)^{\| \vect{k}\|}, 
\end{equation}
where $r_n(\sigma)$ is the number of integer-valued vectors with Euclidean norm $\sqrt \sigma$:
\begin{equation}
 r_n(\sigma) \triangleq \#\left\{ \vect k \in \mathbb Z^n \colon \| \vect k\|^2 = \sigma\right\}
\end{equation}
Assume that vector $X$ has all integer components. 
In 
\begin{equation}
\E{\imath_Y(X)} = D(P_X \| P_Y) + H(X), 
\end{equation}
letting $P_Y = G_{\frac 1 {\E{\|X\|}}}$, we obtain
\begin{align}
H(X) &\leq \E{\imath_Y(X)} \\
&=  (\E{\|X\|} + 1) h \left( \frac 1 {\E{\|X\|} + 1}\right) + \E{\log r_n\left(\|X\|^2\right)}\\
&=  (\E{\|X\|} + 1) h \left( \frac 1 {\E{\|X\|} + 1}\right) + \E{\log r_n\left(\|X\|^2\right)}\\
&\leq \log (\E{\|X\|} +  1) + \log e +  \E{\log r_n\left(\|X\|^2\right)} \label{eq:em1}
\end{align}
To bound the last term in \eqref{eq:em1}, we need an estimate of $r_n\left(s \right)$. 
Towards that end, denote the number of integer-valued vectors within a ball of radius $\sqrt \sigma$ as
\begin{align}
 R_n(\sigma) &\triangleq \sum_{i = 0}^{\sigma} r_{n}(i), 
\end{align}
and observe that
\begin{align}
r_n(\sigma) \leq R_n(\sigma) 
&\leq 1 + \sigma r_n(\sigma)
\end{align}
Furthermore, $R_n(\sigma)$ is upper bounded by the number of cubes with side $1$ that intersect the ball of radius $\sqrt \sigma$. Since the diameter of cube of side $1$ is $\sqrt n$, all those cubes are contained inside the augmented ball of radius $\sqrt \sigma + \sqrt n$. The number of volume-1 cubes that can be packed inside a ball is upper bounded by the volume of that ball; therefore,
\begin{equation}
 R_n(\sigma) \leq b_n \left( \sqrt \sigma + \sqrt n \right)^n \label{eq:Rnsu1}
\end{equation}

A nonasymptotic Stirling's formula leads to
\begin{equation}
\log b_n \leq \frac n 2 \log \frac {2 \pi e} {n + 2} - \frac 1 2 \log (n + 2) + \log \frac {e}{\sqrt \pi}, \label{eq:Anu1}
\end{equation}
 so applying $\frac{z}{1 + z} \leq \log_e (1 + z)$ for $z > -1$ and \eqref{eq:Anu1} in \eqref{eq:Rnsu1}, we obtain
\begin{align}
\log r_n\left(\sigma\right) &\leq  \log b_n - \frac n 2 \log n + n \log \left( \frac{\sqrt \sigma}{\sqrt n} + 1\right) \\
&\leq n \log \left( \frac{\sqrt \sigma}{\sqrt n} + 1\right)  + n \log  \sqrt{2 \pi e} - \frac 1 2 \log (n + 2) - \log \sqrt \pi + \frac {2}{n + 2} \log e \label{eq:rnu1},
\end{align}
and, finally, using \eqref{eq:rnu1} in \eqref{eq:em1} and applying Jensen's inequality, we conclude
\begin{align}
H(X) &\leq  \log (\E{\|X\|} +  1) + n \log \left( \frac 1 {\sqrt n}  \E {\|X\|} + 1\right)  + n \log  \sqrt{2 \pi e} \\
&- \frac 1 2 \log (n + 2) + \log \frac e {\sqrt \pi} + \frac {2}{n + 2} \log e
\end{align}

\begin{remark}
 A classical result in number theory states that  
for $n \geq 5$, \cite[(12.9)]{grosswald2012representations}
\begin{equation}
r_n(\sigma) = S_{n,\sigma} A_{n-1} \sigma^{\frac n 2 - 1}  + C_{n,\sigma} \sigma^{\frac n 4}, \label{eq:number}
\end{equation}
where $S_{n,\sigma}$ is the so-called singular series bounded as $a < S_n(\sigma) < b$, and $C_{n,\sigma}$ is bounded as $|C_{n, \sigma}| \leq c$, where $a, b, c > 0$ are universal constants.  Actually, as $n \to \infty$, $S_n(\sigma) \to 1$. 
\end{remark}

A couple of ALTERNATIVE WAYS TO BOUND $r_n()$. 
Since for $a, b > 0$
\begin{align}
\log (a + b) - \log a \leq \frac b a \log e,
\end{align}
we have
 \begin{align}
\log r_n(\sigma ) &\leq \left( \frac n 2 - 1\right) \log \sigma + \log S_{n,\sigma} + \log A_{n-1} + \frac{C_{n,\sigma}}{S_{n,\sigma} A_{n-1}} \frac 1 {\sigma^{\frac n 4 - 1}}
\end{align}
Alternatively, 
\begin{align}
a^{2k} + b^k \leq (a + \sqrt b)^{2k}
\end{align}
for $k = \frac n 2- 1$, we have
\begin{align}
\log r_n(\sigma ) &\leq \left( n - 2\right) \log \left( \sqrt{\sigma} + \alpha_{n, \sigma} \sqrt{\sigma}^{\frac 1 2 + \frac 1 {n - 2}}\right) +   \log S_{n,\sigma} + \log A_{n-1} 
\end{align}
where
\begin{align}
\alpha_{n, \sigma} = \left|\frac{C_{n,\sigma}}{S_{n,\sigma} A_{n-1}}  \right|^{\frac 2 {n - 2}} \to \infty, \qquad {n \to \infty}
\end{align}
A better estimate of the remainder term in \eqref{eq:number} is required. If we could show that in fact $C_{n,\sigma} = A_{n-1} C$ it would suffice. 
\end{apxonly}

\bibliographystyle{IEEEtran}
\bibliography{../../rateDistortion,../../vk}

\begin{thebibliography}{10}
\providecommand{\url}[1]{#1}
\csname url@samestyle\endcsname
\providecommand{\newblock}{\relax}
\providecommand{\bibinfo}[2]{#2}
\providecommand{\BIBentrySTDinterwordspacing}{\spaceskip=0pt\relax}
\providecommand{\BIBentryALTinterwordstretchfactor}{4}
\providecommand{\BIBentryALTinterwordspacing}{\spaceskip=\fontdimen2\font plus
\BIBentryALTinterwordstretchfactor\fontdimen3\font minus
  \fontdimen4\font\relax}
\providecommand{\BIBforeignlanguage}[2]{{%
\expandafter\ifx\csname l@#1\endcsname\relax
\typeout{** WARNING: IEEEtran.bst: No hyphenation pattern has been}%
\typeout{** loaded for the language `#1'. Using the pattern for}%
\typeout{** the default language instead.}%
\else
\language=\csname l@#1\endcsname
\fi
#2}}
\providecommand{\BIBdecl}{\relax}
\BIBdecl

\bibitem{kostina2015allerton}
V.~Kostina, ``Data compression with low distortion and finite blocklength,'' in
  \emph{Proceedings 53rd Annual Allerton Conference on Communication, Control
  and Computing}, Monticello, IL, Oct. 2015.

\bibitem{kostina2016slb}
------, ``When is {S}hannon's lower bound tight?'' in \emph{Proceedings 54th
  Annual Allerton Conference on Communication, Control and Computing},
  Monticello, IL, Oct. 2016.

\bibitem{kostina2011fixed}
V.~Kostina and S.~Verd\'{u}, ``Fixed-length lossy compression in the finite
  blocklength regime,'' \emph{IEEE Transactions on Information Theory},
  vol.~58, no.~6, pp. 3309--3338, June 2012.

\bibitem{blahut1972computation}
R.~Blahut, ``{Computation of channel capacity and rate-distortion functions},''
  \emph{IEEE Transactions on Information Theory}, vol.~18, no.~4, pp. 460--473,
  Jul. 1972.

\bibitem{shannon1959coding}
C.~E. Shannon, ``{Coding theorems for a discrete source with a fidelity
  criterion},'' \emph{IRE Int. Conv. Rec.}, vol.~7, no.~1, pp. 142--163, Mar.
  1959, reprinted with changes in {\it Information and Decision Processes}, R.
  E. Machol, Ed. New York: McGraw-Hill, 1960, pp. 93-126.

\bibitem{linkov1965evaluation}
Y.~N. Linkov, ``Evaluation of $\epsilon$-entropy of random variables for small
  $\epsilon$,'' \emph{Problems of Information Transmission}, vol.~1, pp.
  18--26, 1965.

\bibitem{linder1994asymptotic}
T.~Linder and R.~Zamir, ``On the asymptotic tightness of the {S}hannon lower
  bound,'' \emph{IEEE Transactions on Information Theory}, vol.~40, no.~6, pp.
  2026--2031, Nov. 1994.

\bibitem{conway1982fast}
J.~H. Conway and N.~Sloane, ``Fast quantizing and decoding and algorithms for
  lattice quantizers and codes,'' \emph{IEEE Transactions on Information
  Theory}, vol.~28, no.~2, pp. 227--232, Mar. 1982.

\bibitem{gersho2012vector}
A.~Gersho and R.~M. Gray, \emph{Vector quantization and signal
  compression}.\hskip 1em plus 0.5em minus 0.4em\relax Springer Science \&
  Business Media, 2012, vol. 159.

\bibitem{renyi1959dimension}
A.~R{\'e}nyi, ``On the dimension and entropy of probability distributions,''
  \emph{Acta Mathematica Academiae Scientiarum Hungarica}, vol.~10, no. 1-2,
  pp. 193--215, Mar. 1959.

\bibitem{csiszar1971entropyquantization}
I.~Csisz\'{a}r, ``Generalized entropy and quantization problems,'' in
  \emph{Trans. Sixth Prague Conf Inform. Theory, Statist. Decision Functions,
  Random Processes}.\hskip 1em plus 0.5em minus 0.4em\relax Prague,
  Czechoslovakia: Akademia, Sep. 1971, pp. 159--174.

\bibitem{rogers1964packing}
C.~A. Rogers, \emph{Packing and covering}.\hskip 1em plus 0.5em minus
  0.4em\relax Cambridge University Press, 1964, no.~54.

\bibitem{gersho1979asymptotically}
A.~Gersho, ``Asymptotically optimal block quantization,'' \emph{IEEE
  Transactions on Information Theory}, vol.~25, no.~4, pp. 373--380, Jul. 1979.

\bibitem{zamir1992universal}
R.~Zamir and M.~Feder, ``On universal quantization by randomized
  uniform/lattice quantizers,'' \emph{IEEE Transactions on Information Theory},
  vol.~38, no.~2, pp. 428--436, Mar. 1992.

\bibitem{linder1994tessellating}
T.~T. Linder and K.~K. Zeger, ``Asymptotic entropy-constrained performance of
  tessellating and universal randomized lattice quantization,'' \emph{IEEE
  Transactions on Information Theory}, vol.~40, no.~2, pp. 575--579, Mar. 1994.

\bibitem{zamir1996onlatticenoise}
R.~Zamir and M.~Feder, ``On lattice quantization noise,'' \emph{IEEE
  Transactions on Information Theory}, vol.~42, no.~4, pp. 1152--1159, Jul.
  1996.

\bibitem{shannon1959probability}
C.~E. Shannon, ``{Probability of error for optimal codes in a Gaussian
  channel},'' \emph{Bell Syst. Tech. J.}, vol.~38, no.~3, pp. 611--656, Mar.
  1959.

\bibitem{bennett1948spectra}
W.~R. Bennett, ``Spectra of quantized signals,'' \emph{Bell Systems Tech. J.},
  vol.~27, pp. 446--472, 1948.

\bibitem{zador1982quantizationerror}
P.~Zador, ``Asymptotic quantization error of continuous signals and the
  quantization dimension,'' \emph{IEEE Transactions on Information Theory},
  vol.~28, no.~2, pp. 139--149, Mar. 1982.

\bibitem{csiszar1974extremum}
I.~Csisz\'{a}r, ``{On an extremum problem of information theory},''
  \emph{Studia Scientiarum Mathematicarum Hungarica}, vol.~9, no.~1, pp.
  57--71, Jan. 1974.

\bibitem{gallager1968information}
R.~Gallager, \emph{{Information theory and reliable communication}}.\hskip 1em
  plus 0.5em minus 0.4em\relax John Wiley \& Sons, Inc. New York, 1968.

\bibitem{berger1971rate}
T.~Berger, \emph{{Rate distortion theory}}.\hskip 1em plus 0.5em minus
  0.4em\relax Prentice-Hall, Englewood Cliffs, NJ, 1971.

\bibitem{pinkston1969application}
J.~Pinkston, ``An application of rate-distortion theory to a converse to the
  coding theorem,'' \emph{IEEE Transactions on Information Theory}, vol.~15,
  no.~1, pp. 66--71, Jan. 1969.

\bibitem{gray1971stationary}
R.~M. Gray, ``Information rates of stationary ergodic finite-alphabet
  sources,'' \emph{IEEE Transactions on Information Theory}, vol.~17, no.~5,
  pp. 516--523, Sep. 1971.

\bibitem{polyanskiy2012notes}
Y.~Polyanskiy, ``6.441: Information theory lecture notes,'' \emph{Dep.
  Electrical Engineering and Computer Science, M.I.T.}, 2012.

\bibitem{palzer2016converse}
L.~Palzer and R.~Timo, ``A converse for lossy source coding in the finite
  blocklength regime,'' in \emph{Proceedings International Zurich Seminar on
  Communications}, Mar. 2016, pp. 15--19.

\bibitem{gerrish1964nongaussian}
A.~Gerrish and P.~Schultheiss, ``Information rates of non-{G}aussian
  processes,'' \emph{IEEE Transactions on Information Theory}, vol.~10, no.~4,
  pp. 265--271, Oct. 1964.

\bibitem{gray1970information}
R.~Gray, ``Information rates of autoregressive processes,'' \emph{IEEE
  Transactions on Information Theory}, vol.~16, no.~4, pp. 412 -- 421, July
  1970.

\bibitem{gray1971markov}
------, ``Rate distortion functions for finite-state finite-alphabet {M}arkov
  sources,'' \emph{IEEE Transactions on Information Theory}, vol.~17, no.~2,
  pp. 127 -- 134, Mar. 1971.

\bibitem{alon1994lower}
N.~Alon and A.~Orlitsky, ``A lower bound on the expected length of one-to-one
  codes,'' \emph{IEEE Transactions on Information Theory}, vol.~40, no.~5, pp.
  1670--1672, Sep. 1994.

\bibitem{wyner1972upper}
A.~Wyner, ``An upper bound on the entropy series,'' \emph{Information and
  Control}, vol.~20, no.~2, pp. 176--181, Mar. 1972.

\bibitem{verdu2009notes}
S.~Verd\'{u}, ``{ELE}528: Information theory lecture notes,'' \emph{Princeton
  University}, 2009.

\bibitem{conway2013sphere}
J.~H. Conway and N.~J.~A. Sloane, \emph{Sphere packings, lattices and
  groups}.\hskip 1em plus 0.5em minus 0.4em\relax Springer Science \& Business
  Media, 2013, vol. 290.

\bibitem{posner1967epsilonentropy}
E.~C. Posner, E.~R. Rodemich, and H.~Rumsey,
  ``\BIBforeignlanguage{English}{Epsilon-entropy of stochastic processes},''
  \emph{\BIBforeignlanguage{English}{The Annals of Mathematical Statistics}},
  vol.~38, no.~4, pp. 1000--1020, Aug. 1967.

\bibitem{cover2012elements}
T.~M. Cover and J.~A. Thomas, \emph{Elements of information theory},
  2nd~ed.\hskip 1em plus 0.5em minus 0.4em\relax John Wiley \& Sons, 2012.

\bibitem{wu2010renyi}
Y.~Wu and S.~Verd{\'u}, ``R{\'e}nyi information dimension: Fundamental limits
  of almost lossless analog compression,'' \emph{IEEE Transactions on
  Information Theory}, vol.~56, no.~8, pp. 3721--3748, Aug. 2010.

\bibitem{koch2015shannonlb}
T.~Koch, ``The {S}hannon lower bound is asymptotically tight,'' \emph{IEEE
  Transactions on Information Theory}, vol.~62, no.~11, pp. 6155--6161, Nov.
  2016.

\bibitem{polyanskiy2015wasserstein}
Y.~Polyanskiy and Y.~Wu, ``Wasserstein continuity of entropy and outer bounds
  for interference channels,'' \emph{IEEE Transactions on Information Theory},
  vol.~62, no.~7, pp. 3992--4002, July 2016.

\bibitem{zamir2014lattice}
R.~Zamir, \emph{Lattice Coding for Signals and Networks: A Structured Coding
  Approach to Quantization, Modulation, and Multiuser Information
  Theory}.\hskip 1em plus 0.5em minus 0.4em\relax Cambridge University Press,
  2014.

\bibitem{koch2015gibbs}
T.~Koch and G.~Vazquez-Vilar, ``Rate-distortion bounds for high-resolution
  vector quantization via {G}ibbs's inequality,'' \emph{arXiv preprint
  arXiv:1507.08349}, 2015.

\bibitem{polyanskiy2010channel}
Y.~Polyanskiy, H.~V. Poor, and S.~Verd{\'u}, ``{Channel coding rate in finite
  blocklength regime},'' \emph{IEEE Transactions on Information Theory},
  vol.~56, no.~5, pp. 2307--2359, May 2010.

\bibitem{fradelizi2015optimal}
M.~Fradelizi, M.~Madiman, and L.~Wang, ``Optimal concentration of information
  content for log-concave densities,'' in \emph{High Dimensional Probability
  {VII}}.\hskip 1em plus 0.5em minus 0.4em\relax Springer, 2016, pp. 45--60.

\end{thebibliography}

\begin{IEEEbiographynophoto}{Victoria Kostina}(S'12--M'14)
joined Caltech as an Assistant Professor of Electrical Engineering in the fall of 2014. She holds a Bachelor's degree from Moscow institute of Physics and Technology (2004), where she was affiliated with the Institute for Information Transmission Problems of the Russian Academy of Sciences, a Master's degree from University of Ottawa (2006), and a PhD from Princeton University (2013). Her PhD dissertation on information-theoretic limits of lossy data compression received Princeton Electrical Engineering Best Dissertation award.  She is also a recipient of Simons-Berkeley research fellowship (2015). Victoria Kostina's research spans information theory, coding, wireless communications and control. 
\end{IEEEbiographynophoto}

\end{document}